\documentclass[aps,reprint,groupedaddress]{revtex4-1}
\usepackage[utf8]{inputenc}
\usepackage[british]{babel}
\usepackage{amsmath}
\usepackage{amsthm}
\usepackage{amsfonts}
\usepackage{amssymb}
\usepackage{mathtools}
\usepackage{tikz}
\usepackage{pgfplots}
\usepackage{hyperref}
\usepackage{hyperref}
\usepackage{amsmath, amssymb, amsfonts}
\usepackage{amssymb,xcolor}
\usepackage{graphicx}
\graphicspath{{images/}}
\usepackage{exscale}
\usepackage{graphicx}
\usepackage{amsmath}
\usepackage{latexsym}
\usepackage{amsfonts}
\usepackage{amssymb}

\def\pmx{\begin{pmatrix}}
\def\emx{\end{pmatrix}}

\newcommand{\ket}[1]{| #1\rangle}
\newcommand{\bra}[1]{\langle #1 |}
\newcommand{\ketbi}[1]{| #1^{AB}\rangle}

\newcommand{\braket}[2]{\langle #1 | #2 \rangle}
\newcommand{\pauli}[1]{\sigma_{#1}}
\newtheorem{theorem}{Theorem}[section]
\newtheorem{proposition}{Proposition}[section]

\theoremstyle{definition}
\newtheorem{define}{Definition}[section]

\begin{document} 

\title{Multipartite entanglement detection for hypergraph states}

\author{M. Ghio$^1$}

\author{D. Malpetti$^2$}

\author{M. Rossi$^3$}

\author{D. Bru{\ss}$^4$}

\author{C. Macchiavello$^5$}

\affiliation{$^1$Scuola Normale Superiore, Piazza dei Cavalieri 7, 56126 Pisa, Italy}
\affiliation{$^2$Laboratoire de Physique, CNRS UMR 5672, Ecole Normale Supérieure de Lyon, Universite de Lyon, 46 Allée d’Italie, Lyon, F-69364, France}
\affiliation{$^3$Dipartimento di Fisica and INFN-Sezione di Pavia, Via Bassi 6, 27100 Pavia, Italy}
\affiliation{$^4$Institut f{\"u}r Theoretische Physik III, Heinrich-Heine-Universit{\"a}t D{\"u}sseldorf, D-40225 D{\"u}sseldorf, Germany}
\affiliation{$^5$Dipartimento di Fisica and INFN-Sezione di Pavia, Via Bassi 6, 27100 Pavia, Italy}

\begin{abstract}

We study the entanglement properties of \textit{quantum hypergraph states} of $n$ qubits, 
focusing on \textit{multipartite entanglement}. We compute multipartite 
entanglement for hypergraph states with a single 
hyperedge of maximum cardinality, for hypergraph states endowed with all 
possible hyperedges of cardinality equal to $n-1$ and for those hypergraph 
states with all possible hyperedges of 
cardinality greater than or equal to $n-1$. We then find a lower bound to the 
multipartite entanglement of a generic quantum hypergraph state.
We finally apply the multipartite entanglement results to the construction 
of \textit{entanglement witness operators}, able to detect genuine 
multipartite entanglement in the neighbourhood of a given hypergraph state. 
We first build entanglement witnesses of the projective type, then propose a 
class of witnesses based on the stabilizer formalism, hence called 
\textit{stabilizer witnesses}, able to reduce the experimental effort from an 
exponential to a linear growth in the number of local measurement settings 
with the number of qubits.

\end{abstract}

\maketitle

\section{Introduction}

Quantum hypergraph states were recently introduced \cite{qu2013,rossi2013}  
 in order to study
a family of multi-qubit entangled states that generalise the notion of graph 
states \cite{graph}, central in various aspects of quantum information, 
such as 
measurement-based quantum computation and quantum error correction. 
This family of states can also be described as locally maximally entangleable 
(LME) states for the particular value $\pi$ of the phase \cite{barbara}. 
Quantum hypergraph states were shown to play a central role in many well known
quantum algorithms \cite{ent-algo,rossi2014} and to provide extreme violation 
of local realism \cite{ottfried}, leading to applications in quantum metrology
and measurement-based quantum computation.

On the other hand, multipartite entanglement is a precious resource in various
quantum information processing tasks, such as for example secret sharing \cite{secret},
multipartite
quantum key distribution \cite{multi-qkd}, distributed dense coding  \cite{dense},
 and some quantum algorithms 
\cite{jl,ent-algo}.
Studying multipartite entanglement properties of quantum states is
therefore of fundamental interest.

In this work we study the multipartite entanglement properties of hypergraph 
states and the possibility of detecting multipartite entanglement via witness 
operators. The paper is organised as follows. In Sect. 
\ref{section_preliminaries} we recall some  notions about 
hypergraph states and multipartite entanglement, that will then be used in the 
rest of the paper. In Sect. \ref{section_multipartite_entanglement} we present
an analytical procedure to derive the multipartite entanglement content for
hypergraph states with a single 
hyperedge of maximum cardinality, for those with all 
possible hyperedges of cardinality equal to $n-1$ and for those 
with all possible hyperedges of 
cardinality greater than or equal to $n-1$. We also derive a lower bound to 
the multipartite entanglement of a generic hypergraph state.
In Sect. \ref{section_entanglement_witnesses} we construct entanglement 
witness operators of two types and analyse their efficiency in terms of number 
of local measurement settings required. We end the paper with a summary
of the results and some concluding remarks in Sect.\ref{Concl}.

\section{Preliminaries}
\label{section_preliminaries}

In this section we define  \textit{quantum graph} and \textit{hypergraph 
states}. We then recall some fundamentals of the theory of \textit{quantum 
entanglement}, introducing \textit{multipartite entanglement}, the 
entanglement measure we will make use of in the rest of our work, and 
focusing on entanglement detection via \textit{entanglement witness operators}.

\subsection{Quantum hypergraph states}

We define \textit{quantum hypergraph states} following the approach of \cite{rossi2013}. For a complete review on graph states we refer to \cite{hein2006}. 

\begin{define}[Hypergraph state - Operational definition]
\label{def_hypergraph_state}
	Let $H=(V,E)$ be a hypergraph of order $n$. To each vertex $i$ we 
associate a qubit $q_{i}$ for $i=1,2,...\,n$, thus associating an $n$-qubit 
quantum system $Q=\lbrace q_{i} \rbrace_{i=1}^{n}$ to the $n$-order 
hypergraph $H$.\\
	We then define the \textit{hypergraph state} $\ket{H}$ associated to 
hypergraph $H$ as the following $n$-qubit pure quantum state
	
\begin{eqnarray}
		\ket{H}:= \prod_{k=1}^{n} \prod_{e\in E, \lvert e \rvert =k} 
C_{k}^{e}\ket{+}^{\otimes n}
\label{def-H}
\end{eqnarray}

where $C_{k}^{e}$ is the $k$-qubit controlled-Z gate acting on the $k$ 
qubits connected by the $k$-hyperedge $e$ and $\ket{+}=\frac{\ket{0}+\ket{1}}
{\sqrt{2}}$ is a superposition of the computational basis states. 
The action of the control gate $C_{k}^{e}$ is defined as

\begin{eqnarray}
C_{k}^{(i_{1},i_{2},... \,i_{k})}=\mathbb{I}^{(j)} \otimes 
(\mathbb{I}-\mathbb{P}) + \pauli{z}^{(j)} \otimes \mathbb{P}
\label{def-Ck}
\end{eqnarray}

for all 
$j=i_{1},i_{2},... \,i_{k}$, where $\mathbb{P}$ is the projector onto the state 
$\ket{11...\,1}^{(i_{1},i_{2},...\,\hat{j},.. \,i_{k})}$ and 
 the notation $\hat{i}$ means that index $i$ is not included. 
Here,  $\sigma_z^{(j)}$ is the Pauli-$z$-operator of vertex $j$.
Hypergraph states with all hyperedges of the same cardinality $k$ are called 
\textit{$k$-uniform}. \textit{Graph states} are a particular case of 
$k$-uniform hypergraph states with $k=2$. 
\end{define}

Hypergraph states, just like graph states, allow for an equivalent definition 
based on a \textit{generalized stabilizer formalism}. However, differently 
from the graphs' \textit{stabilizers} \cite{hein2006}, 
we point out that the \textit{generalized stabilizers} are no more local 
observables.
\begin{define}[Hypergraph state - Stabilizer definition]
\label{def_stabilizer_hypergraph}
	We define the \textit{hypergraph state} $\ket{H}$ associated to the
hypergraph $H$ with $n$ vertices as the unique eigenvector with eigenvalue $1$ of the set of 
$n$ operators $\lbrace K_{i} \rbrace_{i=1}^{n}$ defined as

\begin{eqnarray}
K_{i}:=\pauli{x}^{(i)} \otimes\prod_{k=1}^{n}\prod_{e\in N(i),\lvert e 
\rvert=k-1} C_{k-1}^{e}
\label{def-Ki}
\end{eqnarray}
where  $\sigma_x^{(i)}$ is the Pauli-$x$-operator of vertex $i$
and $N(i)$ denotes the neighbourhood of vertex $i$.

The operators $\lbrace K_{i} \rbrace_{i=1}^{n}$ are called \textit{generalized 
stabilizer operators} of hypergraph state $\ket{H}$; they are hermitian 
operators generating an Abelian group $\Sigma_{n}$ of $2^{n}$ elements 
\cite{rossi2013}. The stabilizers and their 
compositions are hermitian operators.
\end{define}

Given any $n$-qubit hypergraph state $\ket{H}$, we introduce the 
\textit{hypergraph state basis}, generalizing the \textit{graph state basis} 
\cite{hein2006}, with respect to which the stabilizer operators are 
simultaneously diagonalizable.
\begin{proposition}[Hypergraph state basis]
\label{proposition_hypergraph_basis}
	Let $\ket{H}$ be an $n$-qubit hypergraph state and $\lbrace K_{i} 
\rbrace_{i=1}^{n}$ the set of its stabilizer operators. Then the following 
set of $2^{n}$ states

\begin{eqnarray}
	\mathcal{B}_{n}:=\lbrace \ket{\phi_{s}}:=\pauli{z}^{s}\ket{H}\equiv
\pauli{z}^{s_{1}}\otimes \pauli{z}^{s_{2}}\otimes ...\,\pauli{z}^{s_{n}}
\ket{H} \rbrace_{s=0}^{2^{n}-1}
\label{def-Hyperbasis}
\end{eqnarray}
where $s$ is  a binary number composed of bits $s_1,s_2,...s_n$,	
forms a basis for the $n$-qubit Hilbert space $\mathcal{H}_{n}\simeq 
\mathbb{C}^{n}$. Moreover, stabilizer operators $\lbrace K_{i} 
\rbrace_{i=1}^{n}$ are simultaneously diagonalizable with respect to this 
basis

\begin{eqnarray}
		K_{i}\ket{\phi_{s}}=(-1)^{s_{i}} \ket{\phi_{s}}
\label{Ki-rep}
\end{eqnarray}
	
and $\braket{\phi_{s}}{\phi_{t}}=\delta_{s,t}$.
\end{proposition}
\begin{proof}
	We begin by proving the statement in the particular case of an 
$n$-qubit hypergraph state $\ket{H}$ with only one $n$-hyperedge; this allows 
for a very simple representation of its stabilizer operators $K_{i}$ for 
$i=1,2,...\,n$:
	
\begin{eqnarray}
K_{i}=\pauli{x}^{(i)}\otimes C_{n-1}^{(1,2,...\hat{i},...\,n)} \,.
\label{Ki-simrep}
\end{eqnarray}

We first prove that stabilizer operators $K_{i}$ of this form 
commute with the Pauli matrices $\pauli{z}^{(j)}$ whenever $j\neq i$, 
whereas they anticommute when $i=j$. 
	In order to evaluate the action of the stabilizer operator $K_{i}$ on 
qubit $j$ when $i\neq j$, we make use of expression (\ref{def-Ck})
and get

\begin{eqnarray}
K_{i}=\pauli{x}^{(i)} \otimes \mathbb{I}^{(j)} \otimes 
(\mathbb{I}-\mathbb{P}) + \pauli{x}^{(i)}\otimes \pauli{z}^{(j)} 
\otimes \mathbb{P} \,.
\label{Ki-rep2}
\end{eqnarray}
	
Since $\pauli{z}^{(j)}$ commutes both with $\mathbb{I}^{(j)}$ and with 
itself, the commutativity is immediately verified.
	Instead when the two indexes coincide, a negative sign, due to the 
anti-commutativity of the Pauli matrices, appears, namely

\begin{eqnarray}
	\begin{split}
		K_{i}\pauli{z}^{(i)}&=(\pauli{x}^{(i)}
\otimes C_{n-1}^{(1,2,...\hat{i},...\,n)})\pauli{z}^{(i)}\\
		&=-\pauli{z}^{(i)} (\pauli{x}^{(i)}\otimes 
C_{n-1}^{(1,2,...\hat{i},...\,n)})\\
		&=-\pauli{z}^{(i)}  K_{i} \,.
	\end{split}
\label{anti}
\end{eqnarray}
	
It follows that
\begin{eqnarray}
	\begin{split}
		K_{i}\ket{\phi_{s}}&=K_{i}\pauli{z}^{s}\ket{H}\\
		&=(-1)^{s_{i}} \pauli{z}^{s}\ket{H}\\
		&=(-1)^{s_{i}} \ket{\phi_{s}}\,.
	\end{split}
\label{kphi}
\end{eqnarray}

This same reasoning applies to a generic 
stabilizer: it suffices to recall that any stabilizer operator $K_{i}$ may be 
written as the composition of $\pauli{x}^{(i)}\otimes 
\mathbb{I}_{n-1}^{(1,2,... \, \hat{i},...\, n)}$ with $k$-controlled gates of 
the form $\mathbb{I}_{1}^{(i)}\otimes C_{k}^{(i_{1},i_{2},...\, i_{k})}
\otimes \mathbb{I}_{n-k-1}$.\\

	We finally check the orthonormality relation. 
Let $0\leq s,t \leq 
2^{n}-1$ be two different binary numbers $s\neq t$, then there exists at 
least one index $i$ such that $s_{i}\neq t_{i}$, say $s_{i}=1$ and $t_{i}=0$. 
Then $K_{i}\ket{\phi_{s}}=(-1)^{s_{i}} \ket{\phi_{s}}=-\ket{\phi_{s}}$, 
while 
$K_{i}\ket{\phi_{t}}=(-1)^{t_{i}} \ket{\phi_{t}}=\ket{\phi_{t}}$, 
which means 
that $\ket{\phi_{s}}$ and $\ket{\phi_{t}}$ belong to two different, thus 
orthogonal, eigenspaces.
\end{proof}

Just like for the projector on a graph state \cite{hein2006}, as a consequence
of Proposition \ref{proposition_hypergraph_basis}, it can be proved that, 
given a hypergraph state $\ket{H}$, the projector $\ket{H}\bra{H}$ may be 
represented both in terms of the stabilizers $\lbrace K_{i} \rbrace_{i=1}^{n}$
and of the elements of the stabilizer group $\Sigma_{n}$ \cite{guhne2014} as

\begin{equation}
\label{hypegraph_projector}
	\ket{H}\bra{H}=\frac{1}{2^{n}}\sum_{\sigma \in \Sigma_{n}}\sigma = 
\prod_{i=1}^{n}\frac{\mathbb{I}+K_{i}}{2}\,.
\end{equation}

\subsection{Entanglement measures and entanglement detection}

In this work we are interested in \textit{completely} or 
\textit{fully entangled} quantum states of multipartite quantum systems, 
hence in \textit{genuine multipartite entanglement}. 
We remind the reader that the mixed state of a multipartite 
quantum system is said to be completely or fully entangled if it cannot be 
written as a 
convex combination of projectors onto  states that are biseparable
with respect to any bipartition, 
even allowing for different bipartitions in the same decomposition. For a 
complete review of the theory of entanglement and the problem of entanglement 
detection we mainly refer to \cite{guhne2009} and 
\cite{horodecki2009}.
Here, we first study the 
entanglement properties of quantum hypergraph states, and then apply the 
entanglement results to the construction of entanglement witness operators 
for the detection of genuine multipartite entanglement in the neighbourhood 
of a given hypergraph state. 

\begin{define}[Bipartite entanglement - Multipartite entanglement]
	\label{def_bipartite_entanglement}
	Let $\ket{\psi_{n}}\in \mathcal{H}_{n}$ be the pure state of a 
composite quantum system composed of $n$ subsystems $\lbrace 1,2,...\,n 
\rbrace$. Let $AB$ be a possible bipartition of the $n$ subsystems with 
$A=\lbrace 1,2,... \,k \rbrace$ and $B=\lbrace k+1,...\,n\rbrace$ for some 
$1\leq k< n$. We define the \textit{bipartite entanglement} of the state
with respect to bipartition $AB$ as

\begin{eqnarray}
		\begin{split}
			E^{AB}(\ket{\psi_{n}})&:=1-\max_{\ket{\phi^{A}}
\ket{\phi^{B}}}\lvert \bra{\phi^{A}}\braket{\phi^{B}}{\psi_{n}} 
\rvert^{2}\\
			&\equiv 1-\alpha^{AB}(\ket{\psi_n})
		\end{split}
\label{def-bipent}
\end{eqnarray}

where the maximum is taken over all pure biseparable states 
$\ket{\phi^{A}_{k}}\ket{\phi^{B}_{n-k}}$.\\
We define the state's \textit{multipartite entanglement} as its minimum 
bipartite entanglement $E^{AB}(\ket{\phi_{n}})$ with respect to all possible 
bipartitions $AB$:

\begin{eqnarray}
	\begin{split}
			E(\ket{\psi_{n}})&:=\min_{AB}E^{AB}
(\ket{\psi_{n}})\\
			&= 1-\max_{\ket{\phi^{A}}\ket{\phi^{B}},AB}\lvert \bra{\phi^{A}}\braket{\phi^{B}}{\psi_{n}}\rvert^{2}\\
			&\equiv 1-\alpha(\ket{\psi_n})
		\end{split}
\label{def-multient}
\end{eqnarray}

where the maximum is taken over all pure biseparable states 
$\ketbi{\phi}=\ket{\phi^{A}_{k}}\ket{\phi^{B}_{n-k}}$ as well as over all 
possible bipartitions $AB$.
\end{define}
As required by a good measure of entanglement, it can be checked that both bipartite and multipartite entanglement are two \textit{non-increasing} entanglement measures under LOCCs \cite{nielsen1999}.\\
Moreover, in order to compute the overlap between a quantum state $\ket{\psi_{n}}$ and 
the set of all pure biseparable states with respect to a bipartition $AB$, it 
is not necessary to explicitly perform the maximization over the whole set: 
it can be proved \cite{bourennane2004} that

\begin{eqnarray}
\alpha^{AB}(\ket{\psi_n})=\max_{k=1,...\,R}s_{k}^{AB}(\ket{\psi_{n}})^{2}
\label{alpha}
\end{eqnarray}

where $\lbrace s_{k}^{AB}(\ket{\psi_{n}}) \rbrace_{k=1}^{R}$ is the set of 
the Schmidt coefficients of state $\ket{\psi_{n}}$ with respect to 
bipartition $AB$ and $R$ is its Schmidt rank.

\begin{define}[Entanglement witness for genuine multipartite entanglement - 
\cite{horodecki1996}]
\label{def_witness}
Let $\rho_{ent}$ be the density matrix representing a completely entangled 
state of a multipartite quantum system $Q$; let $S_{bi}(Q)$ be the convex 
set of all states that may be written as a convex combination of biseparable 
states. Let $W$ be a hermitian operator such that

\begin{eqnarray}
		\begin{cases}
			\text{Tr}[W\rho_{ent}]<0\\
			\text{Tr}[W\rho_{sep}]\geq 0 \,\,\,\forall \,\,\, 
\rho_{sep}\in S_{bi}(Q) \,.
		\end{cases}
\label{def-W}
\end{eqnarray}
Then operator $W$ is an \textit{entanglement witness} for genuine 
multipartite entanglement.
\end{define}
A standard procedure for the construction of an entanglement witness that is able to 
detect genuine multipartite entanglement in the neighbourhood of a given state
$\ket{H}$ is that of the \textit{projector-based entanglement witness}
\cite{guhne2009}
\begin{eqnarray}
	W:=\alpha(\ket{H}) \mathbb{I}-\ket{H}\bra{H}
\label{proj-W}
\end{eqnarray}
where 
$\alpha(\ket{H})$  is defined in Eq. (\ref{def-multient}).

\section{Multipartite entanglement in quantum hypergraph states}
\label{section_multipartite_entanglement}

In this section we first compute the exact multipartite entanglement formulas 
for some specific classes of hypergraph states and we then derive a 
lower bound to the multipartite entanglement of a generic hypergraph state.

We propose a procedure to evaluate exactly the multipartite entanglement for 
some symmetric classes of hypergraph states. In order to do this we use
the concept of infinity norm, defined as follows.

\begin{define}[Infinity norm]
\label{def_infinity_norm}
Let $M\in \mathbb{C}^{n\times n}$ be a square matrix. 
Its \textit{infinity norm} $\parallel\hspace{-0.1cm} M\hspace{-0.1cm} 
\parallel_{\infty}$ is defined as
\begin{eqnarray}
	\parallel\hspace{-0.1cm} M\hspace{-0.1cm} \parallel_{\infty}:=
\max_{i=1,2,...\,n}\sum_{j=1}^{n}\lvert M_{ij} \rvert\,.
\label{def-inftynorm}
\end{eqnarray}

For $M\geq 0$, its maximum eigenvalue is bounded by $\lambda_{max}(M)\leq 
\parallel\hspace{-0.1cm} M\hspace{-0.1cm} \parallel_{\infty}$ (see for 
instance \cite{horn2012}).
\end{define}

Our procedure which allows to compute $\alpha^{AB}(\ket{\phi})$ in Eq. (\ref{def-multient}) 
is summarized as follows:

\begin{itemize}

\item Take an $n$-qubit hypergraph state $\ket{H_{n}}$ which is invariant 
under permutations of the qubits;

\item Consider the bipartition $\bar{A}=\lbrace 1,2,...\,(n-1) \rbrace$ and 
$\bar{B}=\lbrace n \rbrace$, perform the Schmidt decomposition with respect 
to this bipartition and spot the 
maximum Schmidt coefficient $s_{max}^{\bar{A}\bar{B}}(\ket{H_{n}})$;

\item Consider now all other bipartitions $A=\lbrace 1,2,...\,n-k \rbrace$ 
and $B=\lbrace n-k+1,...\, n \rbrace$ for $k>1$ and write the reduced density 
matrix $\rho^{(12...n-k)}$ corresponding to $n-k$ qubits;

\item Compute the infinity norm  $\parallel \hspace{-1mm}\rho^{(12...n-k)} 
\hspace{-1mm}\parallel_{\infty}$;

\item Compare the infinity norm $\parallel \hspace{-1mm} \rho^{(12...n-k)} 
\hspace{-1mm}\parallel_{\infty}$ with $\left(s_{max}^{\bar{A}\bar{B}}
(\ket{H_{n}})\right)^{2}$;

\item If $\parallel \hspace{-1mm}\rho^{(12...n-k)} 
\hspace{-1mm}\parallel_{\infty}\leq \left(s_{max}^{\bar{A}\bar{B}}
(\ket{H_{n}})\right)^{2}$ for all values $1<k\leq n/2$ then 
$\alpha(\ket{H_{n}})=\left(s_{max}^{\bar{A}\bar{B}}(\ket{H_{n}})\right)^{2}$.
This last step is justified by (\ref{alpha}).

\end{itemize}

In the following we apply this procedure to compute $\alpha(\ket{H_{n}})$, 
and therefore the multipartite entanglement $E(\ket{\phi_{n}})$ via 
(\ref{def-multient}), for some classes of hypergraph states.

\subsection{Hypergraph states with one maximum-cardinality hyperedge}

We here consider $n$-qubit hypergraph states $\ket{G_{n}}$ with only one 
maximum-cardinality $n$-hyperedge, namely
\begin{equation}
	\ket{G_{n}}=C_{n}^{(1,2,...\,n)}\ket{+}^{\otimes n}\,.
\label{def-Gn}
\end{equation}
\begin{theorem}[Multipartite entanglement - One maximum-cardinality hyperedge]
\label{theorem_multipartite_entanglement}
	Let $\ket{G_{n}}$ be an $n$-qubit hypergraph state with just one 
maximum-cardinality $n$-hyperedge. Then the maximum squared overlap between 
hypergraph state $\ket{G_{n}}$ and the pure biseparable states is 

\begin{equation}
	\alpha_{n} = \max_{\ket{\phi^{A}}\ket{\phi^{B}},\lbrace A,B\rbrace} 
|\bra{\phi^{A}}\braket{\phi^{B}}{G_{n}}|^{2}=\frac{2^{n-1}-1}{2^{n-1}}
\label{alpha-Gn}
\end{equation}
and the multipartite entanglement of hypergraph state $\ket{G_{n}}$ is
\begin{eqnarray}
		E(\ket{G_{n}})=\frac{1}{2^{n-1}} \,.
\label{E-Gn}
\end{eqnarray}
\end{theorem}

Hypergraph states with only one maximum-cardinality 
hyperedge are superpositions 
of all the elements of the computational basis with only one negative sign in 
front of the element $\ket{11...\,1}$. These are exactly the same states 
employed by Grover's quantum search algorithm in the single solution case 
\cite{grover}. This result was in fact first proved in Ref. 
\cite{rossi_bruss2013}, where the entanglement dynamics in Grover's 
algorithm is analysed. Here we prove it by following the procedure outlined
above.

\begin{proof}
Consider first the bipartition $\bar{A}=\lbrace 1,2,...\,(n-1) \rbrace$ and 
$\bar{B}=\lbrace n \rbrace$. The Schmidt decomposition of hypergraph state 
$\ket{G_{n}}$ with respect to bipartition $\bar{A}\bar{B}$ is

\begin{eqnarray}
\begin{split}
	\ket{G_{n}}=&\sqrt{\frac{2^{n-1}-1}{2^{n-1}}}\,\,
\sum_{x=0}^{2^{n-1}-2} \frac{\ket{x}}{\sqrt{2^{n-1}-1}}\,\ket{+}\,\,+\\
	&+\frac{\ket{11...1}\ket{-}}{\sqrt{2^{n-1}}}
\end{split}
\label{Gn-bip}
\end{eqnarray}
and the maximum Schmidt coefficient is therefore 
$s_{max}^{\bar{A}\bar{B}}(\ket{G_{n}})=\sqrt{\frac{2^{n-1}-1}{2^{n-1}}}$.\\
Consider now bipartitions $A=\lbrace 1,2,...\,n-k \rbrace$ and $B=\lbrace 
n-k+1,...\, n \rbrace$ with $k>1$. By performing the partial trace over 
the last $k$ subsystems, the reduced density matrix $\rho^{(12...n-k)}$ becomes
\begin{eqnarray}
\footnotesize
	\hspace{-1mm}\frac{1}{2^{n-1}}\hspace{-1mm}
	\left( \begin{matrix}
	2^{k-1} & 2^{k-1} & \dots & 2^{k-1} & 2^{k-1}-1\\
	2^{k-1} & 2^{k-1} & \dots & 2^{k-1} & 2^{k-1}-1\\
	\vdots & \vdots  & \ddots & \vdots  & \vdots \\
	2^{k-1} & 2^{k-1} & \dots & 2^{k-1} & 2^{k-1}-1\\
	2^{k-1}-1 & 2^{k-1}-1 & \dots & 2^{k-1}-1 & 2^{k-1}
	\end{matrix} \right)
\label{rho-Gn}
\end{eqnarray}
Regarding  the maximum eigenvalue, it follows that
\begin{eqnarray}
	\lambda_{max}(\rho^{(12...n-k)})\leq \parallel \hspace{-1mm}
\rho^{(12...n-k)}\hspace{-1mm} \parallel_{\infty}=\frac{2^{n-1}-1}{2^{n-1}}
\label{l-Gn}
\end{eqnarray}
for all $k>1$. 
We then conclude  that $\alpha_{n}=\frac{2^{n-1}-1}{2^{n-1}}$.
\end{proof}

\subsection{Hypergraph states with all ($n$-$1$)-hyperedges}

We now consider $n$-qubit hypergraph states $\ket{H_{n}^{n-1}}$ endowed with 
all possible hyperedges of cardinality $n$-$1$, namely
\begin{equation}
	\ket{H_{n}^{n-1}}=\prod_{i=1}^{n} C_{n-1}^{(1,2,...\,
\hat{i},...\,n)}\ket{+}^{\otimes n}\,.
\label{Hn-1}
\end{equation}

\begin{theorem}[Multipartite entanglement - Hyperedges of cardinality $n$-$1$]
\label{theorem_multipartite_entanglement_n-1}
	Let $\ket{H_{n}^{n-1}}$ be an $n$-qubit hypergraph state endowed with 
all possible hyperedges of cardinality $n$-$1$. 
Then the maximum squared overlap 
between hypergraph state $\ket{H_{n}^{n-1}}$ and the pure biseparable states is
\begin{eqnarray}
		\begin{dcases}
			\alpha(\ket{H_{4}^{3}})=\frac{3+\sqrt{5}}{8}\leq 
\frac{3}{4}\\
			\alpha(\ket{H_{n}^{n-1}})=\frac{2^{n-1}-n}{2^{n-1}}\,
\,\,\mathrm{for}\,\,n\,\mathrm{even},\,\,n\geq 6\\
			\alpha(\ket{H_{n}^{n-1}})=\frac{2^{n-1}-n+1}{2^{n-1}}
\,\,\,\mathrm{for}\,\,n\,\mathrm{odd}\,.
		\end{dcases}
\label{alpha-Hn-1}
\end{eqnarray}

The multipartite entanglement of hypergraph state $\ket{H_{n}^{n-1}}$ is then 
given by
\begin{eqnarray}
		\begin{dcases}
			E(\ket{H_{4}^{3}})=\frac{5-\sqrt{5}}{8}\geq 
\frac{1}{4}\\
			E(\ket{H_{n}^{n-1}})=\frac{n}{2^{n-1}}\,\,\,
\mathrm{for}\,\,n\,\mathrm{even},\,\,n\geq 6\\
			E(\ket{H_{n}^{n-1}})=\frac{n-1}{2^{n-1}}\,\,\,
\mathrm{for}\,\,n\,\mathrm{odd}\,.
		\end{dcases}
\label{E-Hn-1}
\end{eqnarray}
\end{theorem}

The complete proof of this result is reported in Appendix 
\ref{appendix_multipartite_entanglement}. The procedure is a straightforward 
generalization of the one applied to the single maximum-cardinality hyperedge. 
The maximum eigenvalues of the reduced density matrices do not increase for
increasing $k$, hence they remain lower than or equal to the squared maximum 
Schmidt coefficient with respect to the first bipartition 
$\bar{A} \bar{B}$. 
The only exception to this behaviour is the case $n=4$, that is the lowest 
possible even value. We distinguish the case of $n$ even from $n$ odd 
because of a 
difference in the sign of the coefficient in front of the computational basis 
element $\ket{11...\,1}$. While a hypergraph state $\ket{H_{n}^{n-1}}$ with 
$n$ even has $n$ negative coefficients, a hypergraph state 
$\ket{H_{n}^{n-1}}$ with $n$ odd has an additional negative sign in front of 
the component $\ket{11...\,1}$: when $n$ is even the negative signs 
introduced by the controlled-$Z$ gates compensate each other.

\subsection{Hypergraph states with all hyperedges of cardinality greater than or equal to $n$-$1$}
We here consider $n$-qubit hypergraph states endowed with all possible 
hyperedges of cardinality greater than or equal to $n$-$1$, namely
\begin{equation}
	\ket{H_{n}^{n-1,n}}=C_{n}^{(1,2,...\,n)}\prod_{i=1}^{n} 
C_{n-1}^{(1,2,...\,\hat{i},...\,n)}\ket{+}^{\otimes n}\,.
\end{equation}
\begin{theorem}[Multipartite entanglement -  Hyperedges of cardinality 
greater than or equal to $n$-$1$]
\label{theorem_multipartite_entanglement_n-1n}
	Let $\ket{H_{n}^{n-1,n}}$ be an $n$-qubit hypergraph state endowed 
with all possible hyperedges of cardinality greater than or equal to $n$-$1$. 
Then the maximum squared overlap between hypergraph state 
$\ket{H_{n}^{n-1,n}}$ and the pure biseparable states is
\begin{eqnarray}
		\begin{dcases}
			\alpha(\ket{H_{3}^{2,3}})=\frac{3}{4}\\
			\alpha(\ket{H_{n}^{n-1,n}})=\frac{2^{n-1}-n+1}
{2^{n-1}}\,\,\,\mathrm{for}\,\,n\,\mathrm{even}\\
			\alpha(\ket{H_{n}^{n-1,n}})=\frac{2^{n-1}-n}
{2^{n-1}}\,\,\,\mathrm{for}\,\,n\,\mathrm{odd},\,\,n\geq 5\,.
		\end{dcases}
\label{alpha-Hn,n-1}
\end{eqnarray}
The multipartite entanglement of hypergraph state $\ket{H_{n}^{n-1,n}}$ is
\begin{eqnarray}
		\begin{dcases}
			E(\ket{H_{3}^{2,3}})=\frac{1}{4}\\
			E(\ket{H_{n}^{n-1,n}})=\frac{n-1}{2^{n-1}}\,\,\,
\mathrm{for}\,\,n\,\mathrm{even}\\
			E(\ket{H_{n}^{n-1,n}})=\frac{n}{2^{n-1}}\,\,\,
\mathrm{for}\,\,n\,\mathrm{odd},\,\,n\geq 5\,.
		\end{dcases}
\label{E-Hn,n-1}
\end{eqnarray}
\end{theorem}
The complete proof of this result is reported in Appendix 
\ref{appendix_multipartite_entanglement}. The procedure is a straightforward 
generalization of the one applied to the single maximum-cardinality 
hyperedge case 
and it is analogous to the procedure applied to prove Theorem 
\ref{theorem_multipartite_entanglement_n-1}, it only differs 
in the opposite role played by the parity of  the number of qubits $n$. 
The additional $n$-hyperedge, with respect to the previous case, changes the 
sign of the coefficient in front of the component $\ket{11...\,1}$. 
We here hence distinguish the case of $n$ even from case of $n$ odd, just 
like we
did in the previous case, perform the same demonstrative procedure and find 
inverted formulas between the two cases of $n$ even and $n$ odd. As expected, 
the only exception here is for $n=3$, that is the lowest possible value for 
$n$ odd, similarly to the previous exception of the case $n=4$ 
(the lowest possible value for $n$ under the hypothesis of 
Theorem \ref{theorem_multipartite_entanglement_n-1}).

\subsection{Lower bound to the multipartite entanglement of a generic hypergraph state}
\begin{theorem}[Multipartite entanglement - General case]
\label{theorem_multipartite_entanglement_general}
	Let $\ket{H_{n}^{k_{max}}}$ be an $n$-qubit connected hypergraph state of maximum hyperedge-cardinality equal to $k_{max}$. Then its overlap with the pure biseparable states is upper bounded by
\begin{eqnarray}
		\alpha(\ket{H_{n}^{k_{max}}}) \leq \frac{2^{k_{max}-1}-1}{2^{k_{max}-1}}\,.
\label{alpha-gen}
\end{eqnarray}
Its multipartite entanglement is hence lower bounded by
\begin{eqnarray}
		E(\ket{H_{n}^{k_{max}}}) \geq \frac{1}{2^{k_{max}-1}}\,.
\label{E-gen}
\end{eqnarray}
\end{theorem}

\begin{proof}
The proof of the theorem may be outlined as follows.
\begin{itemize}
\item Given an $n$-qubit connected hypergraph state $\ket{H_{n}^{k_{max}}}$ of maximum hyperedge-cardinality equal to $k_{max}$, we consider a possible bipartition $AB$. Among the hyperedges that cross the bipartition we choose
one with the highest cardinality, which we denote as $\kappa$: by definition $\kappa \leq k_{max}$. The reason why we choose an hyperedge with the highest cardinality  will be made clear in the next steps.\\
\item  We show that hypergraph state $\ket{H_{n}^{k_{max}}}$ may always be reduced to a mixture of single-hyperedge hypergraph states $\ket{G_{\kappa'}}$ with $\kappa'\leq\kappa \leq k_{max}$ by only means of operations that are local with respect to the chosen bipartition $AB$.\\
\item Given the non-increasing property of bipartite entanglement under LOCCs, the entanglment of the initial state is greater than or equal to the weighted average of the entanglement values of the single states belonging to the mixture, with weights given by the probabilities of the measurement outcomes. In particular the initial entanglement is greater than or equal to the minimum value entering the weighted average. We deduce that
\begin{displaymath}
	E^{AB}(\ket{H_{n}^{k_{max}}})\geq E^{AB}(\ket{G_{\kappa'}})
\end{displaymath}
where $\kappa'$ is the maximum cardinality within the above mixture.\\
\item Recalling that multipartite entanglement is defined as the minimum of bipartite entanglement over all possible bipartitions and applying Theorem \ref{theorem_multipartite_entanglement}, it follows that
\begin{displaymath}
	E^{AB}(\ket{H_{n}^{k_{max}}})\geq E(\ket{G_{\kappa'}})=\frac{1}{2^{\kappa'-1}}
\end{displaymath}
where $E(\ket{G_{\kappa'}})$ denotes the minimum of $E^{AB}(\ket{G_{\kappa'}})$ over all possible bipartitions AB.\\
\item We conclude by observing that the minimum value of the multipartite entanglement is attained when $\kappa=\kappa'=k_{max}$. In general $\kappa'\leq \kappa\leq k_{max}$ but it is possible that $\kappa' =k_{max}$ if the initially considered bipartition crosses a $k_{max}$-hyperedge. In general if $\kappa_{1}\leq\kappa_{2}$ then $\frac{1}{2^{\kappa_{1}-1}}\geq\frac{1}{2^{\kappa_{2}-1}}$ and, since we are looking for the minimum, this motivates the choice of the hyperedge with highest cardinality at the first step. This leads to
\begin{displaymath}
	E(\ket{H_{n}^{k_{max}}})\geq E(\ket{G_{k_{max}}})=\frac{1}{2^{k_{max}-1}}\,.
\end{displaymath}
\end{itemize}

In order to complete the proof it hence suffices to show how to reduce $\ket{H_{n}^{k_{max}}}$ to a single-hyperedge hypergraph state $\ket{G_{\kappa'}}$ with $\kappa'\leq k_{max}$ by only means of operations that are local with respect to the chosen bipartition $AB$ and single-qubit measurements. This may be achieved through the following iterative procedure (see Fig.\ref{fig_lower_bound} for an example).
\begin{itemize}
\item Given a bipartition $AB$, choose one of the hyperedges with the highest cardinality crossed by the bipartition and call $\kappa$ its cardinality.\\
\item Perform $\pauli{z}$ measurements on the $n-\kappa$ qubits not belonging to the chosen hyperedge (Fig.\ref{fig_lower_bound}, step $1$). The resulting state will be of the form
\begin{displaymath}
	\ket{H_{\kappa}}\ket{\phi^{(1)}}\ket{\phi^{(2)}}...\,\ket{\phi^{(n-\kappa)}}
\end{displaymath}
where $\ket{H_{\kappa}}$ is a $\kappa$-qubit hypergraph state with an hyperedge with highest cardinality $\kappa$ and possibly other internal lower-cardinality hyperedges; state $\ket{\phi^{(i)}}\in \lbrace \ket{0},\ket{1}\rbrace$ is the single-qubit state of the qubit at vertex $i$ and depends on the corresponding measurement output. These measurements do not delete the chosen $\kappa$-hyperedge but may cause the appearance of internal lower-cardinality hyperedges \cite{qu2013}.\\
\item Remove all internal hyperedges of cardinality $\kappa-1$ by means of local Pauli operations \cite{guhne2014} (Fig.\ref{fig_lower_bound}, step $2$). This may introduce edges of lower cardinality that in general may not be removed by only means of LOCCs.\\
\item Remove all hyperedges of cardinality $k'<\kappa-1$ that do not cross the chosen bipartition by means of controlled gates of the form $C_{k'}$; even if these are not single-qubit transformations they are local with respect to the chosen bipartition. The non-increasing property of bipartite entanglement under LOCCs therefore applies to this case as well.\\
\item Stop if at this stage all lower-cardinality hyperedges have been removed, i.e. the initial state has been reduced to a state of the form $\ket{G_{\kappa}}\ket{\phi^{(1)}}\ket{\phi^{(2)}}...\,\ket{\phi^{(n-\kappa)}}$ (Fig.\ref{fig_lower_bound}, step $3$ left).\\
\item If this is not the case (Fig.\ref{fig_lower_bound}, step $3$ right), it means that the remaining state is still of the form $\ket{H_{\kappa}}\ket{\phi^{(1)}}\ket{\phi^{(2)}}...\,\ket{\phi^{(n-\kappa)}}$ where $\ket{H_{\kappa}}$ is a $\kappa$-qubit hypergraph state with an hyperedge with highest cardinality $\kappa$ and possibly other internal lower-cardinality hyperedges. Consider then the lower-cardinality hyperedges that remain: they all cross the bipartition because those not crossing the bipartition were removed in the previous steps, moreover they are all of cardinality strictly lower than $\kappa-1$. Select an hyperedge with highest cardinality and denote its cardinality with $\tilde{\kappa}$. Measure one of the qubits outside the $\tilde{\kappa}$-hyperedge but still within the $\kappa$-hyperedge; this may cause, depending on the measurement outcome, the appearance of a ($\kappa$-$1$)-hyperedge crossing the bipartition. Select now again an hyperedge with the highest cardinality among those crossing the bipartition, call $\kappa'$ its cardinality, repeat the procedure from the beginning replacing $\kappa$ with $\kappa'$.
\end{itemize}
\end{proof}

\begin{figure}[h!]
	\centering
	\begin{tikzpicture}[scale=1.5]
		\path (-1,0) node[label=right:\(q_3\)] (BL_q3) {$\bullet$}; \path (-2,0) node[label=left:\(q_4\)] (BL_q4) {$\bullet$};
		\path (-1,0.5) node[label=right:\(q_2\)] (BL_q2) {$\bullet$}; \path (-2,0.5) node[label=left:\(q_5\)] (BL_q5) {$\bullet$};
		\path (-1.5,1) node[label=below:\(q_1\)] (BL_q1) {$\bullet$};
		\draw [thick, black] (BL_q3.center)--(BL_q4.center);
		\path (-1.9,1.3) node[] (BLa) {};\path (-1.2,-0.4) node[] (BLb) {};\draw [thick, red] (BLa)--(BLb);
		\path (1,0) node[label=left:\(q_4\)] (BR_q4) {$\bullet$}; \path (2,0) node[label=right:\(q_3\)] (BR_q3) {$\bullet$};
		\path (1,0.5) node[label=left:\(q_5\)] (BR_q5) {$\bullet$}; \path (2,0.5) node[label=right:\(q_2\)] (BR_q2) {$\bullet$};
		\path (1.5,1) node[label=below:\(q_1\)] (BR_q1) {$\bullet$};
		\draw [thick, black] (BR_q3.center)--(BR_q4.center);
		\draw [thick,green] plot [smooth cycle] coordinates {(2.1,0)(1,-0.05)(1,0.6)};
		\path (1.8,-0.4) node[] (BRa) {};\path (1.1,1.3) node[] (BRb) {};\draw [thick, red] (BRa)--(BRb);
		\path (-0.5,2) node[label=left:\(q_4\)] (BM_q4) {$\bullet$}; \path (0.5,2) node[label=right:\(q_3\)] (BM_q3) {$\bullet$};
		\path (-0.5,2.5) node[label=left:\(q_5\)] (BM_q5) {$\bullet$}; \path (0.5,2.5) node[label=right:\(q_2\)] (BM_q2) {$\bullet$};
		\path (0,3) node[label=below:\(q_1\)] (BM_q1) {$\bullet$};
		\draw [thick,blue] plot [smooth cycle] coordinates {(0.6,2.6)(0.6,1.9)(-0.6,1.9)(-0.6,2.6)};
		\draw [thick, black] (BM_q3.center)--(BM_q4.center);
		\path (-0.4,3.3) node[] (TMa) {};\path (0.3,1.5) node[] (TMb) {};\draw [thick, red] (TMa)--(TMb);
		\path (-0.5,6) node[label=left:\(q_4\)] (TM_q4) {$\bullet$}; \path (0.5,6) node[label=right:\(q_3\)] (TM_q3) {$\bullet$};
		\path (-0.5,6.5) node[label=left:\(q_5\)] (TM_q5) {$\bullet$}; \path (0.5,6.5) node[label=right:\(q_2\)] (TM_q2) {$\bullet$};
		\path (0,7) node[label=below:\(q_1\)] (TM_q1) {$\bullet$};
		\draw [thick,blue] plot [smooth cycle] coordinates {(0.6,6.6)(0.6,5.9)(-0.6,5.9)(-0.6,6.6)};
		\draw [thick, black] (TM_q3.center)--(TM_q4.center);
		\draw [thick, black] (TM_q1.center)--(TM_q2.center);
		\draw [thick,green] plot [smooth cycle] coordinates {(0.6,6)(-0.55,5.95)(-0.55,6.6)};
		\path (-0.4,7.3) node[] (TMa) {};\path (0.3,5.5) node[] (TMb) {};\draw [thick, red] (TMa)--(TMb);
		\path (1,4) node[label=left:\(q_4\)] (TR_q4) {$\bullet$}; \path (2,4) node[label=right:\(q_3\)] (TR_q3) {$\bullet$};
		\path (1,4.5) node[label=left:\(q_5\)] (TR_q5) {$\bullet$}; \path (2,4.5) node[label=right:\(q_2\)] (TR_q2) {$\bullet$};
		\path (1.5,5) node[label=below:\(q_1\)] (TR_q1) {$\bullet$};
		\draw [thick,blue] plot [smooth cycle] coordinates {(2.1,4.6)(2.1,3.9)(0.9,3.9)(0.9,4.6)};
		\draw [thick, black] (TR_q3.center)--(TR_q4.center);
		\draw [thick,green] plot [smooth cycle] coordinates {(2.1,4)(1,3.95)(1,4.6)};
		\draw [thick, violet] (TR_q2) circle (1.2mm);
		\path (1.8,3.5) node[] (TRa) {};\path (1.1,5.3) node[] (TRb) {};\draw [thick, red] (TRa)--(TRb);
		\path (-1,4) node[label=right:\(q_3\)] (TL_q3) {$\bullet$}; \path (-2,4) node[label=left:\(q_4\)] (TL_q4) {$\bullet$};
		\path (-1,4.5) node[label=right:\(q_2\)] (TL_q2) {$\bullet$}; \path (-2,4.5) node[label=left:\(q_5\)] (TL_q5) {$\bullet$};
		\path (-1.5,5) node[label=below:\(q_1\)] (TL_q1) {$\bullet$};
		\draw [thick,blue] plot [smooth cycle] coordinates {(-0.9,4.6)(-0.9,3.9)(-2.1,3.9)(-2.1,4.6)};
		\draw [thick, black] (TL_q3.center)--(TL_q4.center);
		\draw [thick,green] plot [smooth cycle] coordinates {(-0.9,4)(-2,3.95)(-2,4.6)};
		\path (-1.9,5.3) node[] (TLa) {};\path (-1.2,3.5) node[] (TLb) {};\draw [thick, red] (TLa)--(TLb);
		\draw [->,thick,black] (TM_q4)--(TL_q1);\draw [->,thick,black] (TM_q3)--(TR_q1);
		\draw [->,thick,black] (-1.5,3.8)--(BM_q5); \draw [->,thick,black] (1.5,3.8)--(BM_q2);
		\draw [->,thick,black] (BM_q4)--(BL_q1);\draw [->,thick,black] (BM_q3)--(BR_q1);
		\path (-1.5,5.5) node[] (L1) {$(M_{\pauli{z}^{(1)}},0)$}; \path (1.5,5.5) node[] (L2) {$(M_{\pauli{z}^{(1)}},1)$};
		\path (-1.5,3) node[] (L3) {$\pauli{x}^{(2)}$}; \path (1.5,3) node[] (L4) {$\pauli{z}^{(2)}\,\pauli{x}^{(2)}$};
		\path (-1.5,1.5) node[] (L5) {$(M_{\pauli{z}^{(2)}},0)$}; \path (1.5,1.5) node[] (L6) {$(M_{\pauli{z}^{(2)}},1)$};
		\path (1,-2) node[label=left:\(q_5\)] (BBR_q5) {$\bullet$}; \path (2,-2) node[label=right:\(q_2\)] (BBBR_q2) {$\bullet$};
		\path (1,-2.5) node[label=left:\(q_4\)] (BBR_q4) {$\bullet$}; \path (2,-2.5) node[label=right:\(q_3\)] (BBR_q3) {$\bullet$};
		\path (1.5,-1.5) node[label=below:\(q_1\)] (BBR_q1) {$\bullet$};
		\draw [thick,green] plot [smooth cycle] coordinates {(2.1,-2.5)(1,-2.55)(1,-1.9)};
		\draw [->,thick,black] (1.5,-0.5)--(BBR_q1); \path(1.8,-1) node[] (L7) {$\pauli{x}^{(5)}$};
		\path (1.8,-2.9) node[] (BBRa) {};\path (1.1,-1.2) node[] (BBRb) {};\draw [thick, red] (BBRa)--(BBRb);
	\end{tikzpicture}
	\caption{Procedure to transform a $5$-qubit hypergraph state of maximum hyperedge cardinality equal to $4$ probabilistically into a single-hyperedge hypergraph state by only means of transformations that are local with respect to a chosen bipartition. The exemplifying bipartition is $A=\lbrace 1,2,3 \rbrace$ and $B=\lbrace 4,5\rbrace$ (red inclined line). $(M_{\pauli{z}^{(i)}},j)$ denotes a $\pauli{z}$ measurement to be performed on qubit $i$ with outcome $j$. Depending on the measurement outcomes the output state may be either of the form $\ket{G_{2}}\ket{\phi^{(1)}}\ket{\phi^{(2)}}\ket{\phi^{(5)}}$ or of the form $\ket{G_{3}}\ket{\phi^{(1)}}\ket{\phi^{(2)}}$. This leads to a lower bound of $\frac{1}{4}$ for the bipartite entanglement with respect to this choice of the bipartition.}
	\label{fig_lower_bound}
\end{figure}
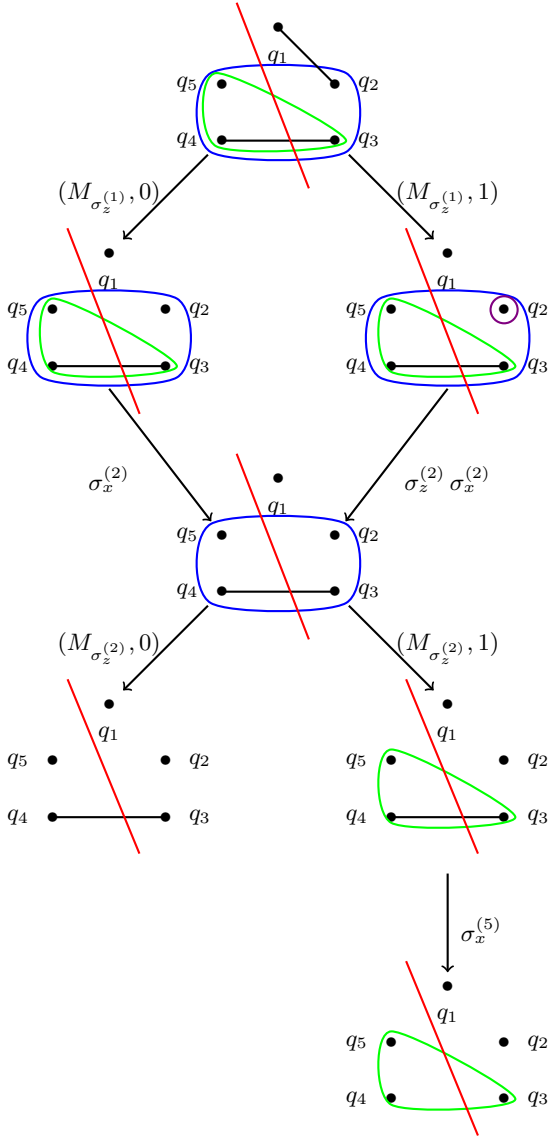

\section{Entanglement witnesses in the hypergraph formalism}
\label{section_entanglement_witnesses}

In this section we apply the multipartite entanglement results of Section 
\ref{section_multipartite_entanglement} to the construction of 
\textit{entanglement witnesses}. We first derive entanglement witnesses of 
the projective type, then propose a class of witnesses based on the 
stabilizer formalism, hence called \textit{stabilizer witnesses}, requiring a 
lower number of measurement settings than the projective ones. The efficiency 
of the constructed witnesses is evaluated on the basis of their robustness to 
noise and of the number of local measurement settings required by each of 
them in order to be measured.\\

Let $W_{n}$ be an entanglement witness able to detect entanglement in the 
neighbourhood of the $n$-qubit hypergraph state $\ket{H_{n}}$. Let $R_{p}$ be 
the hypergraph state $\ket{H_{n}}$ after the action of some white noise, 
i.e. $R_{p}=p\frac{\mathbb{I}}{2^{n}}+(1-p)\ket{H_{n}}\bra{H_{n}}$ with 
$0\leq p \leq 1$. We define as \textit{robustness parameter} the limit value 
$p^{L}_{n}$ for $p$, such that $Tr[R_{p}W_{n}]<0$ for all $p<p^{L}_{n}$. Note 
that $p^{L}_{n}$ also quantifies the dimensions of neighbourhood where 
$W_{n}$ is able to detect entanglement.\\
We remind that a generic witness operator can be decomposed in terms of a 
set of
local observables $\mathcal{O}=\lbrace O^{(i)} \rbrace_{i=1}^{m\leq n}$, 
i.e. $O^{(i)}$ acts on qubit $i$ for $i=1,2,...1,m\leq n$ \cite{ourwit}. We 
say that we  measure the \textit{local measurement setting} $\mathcal{O}$ 
if we perform  the simultaneous von Neumann measurement of the observables 
in $\mathcal{O}$. The evaluation of the number of local measurement settings
required to measure the expectation value of each witness operator
is reported in Appendix \ref{appendix_measurement}.

\subsection{Projector-based entanglement witnesses}

In the following we list the projector-based entanglement witnesses that we
constructed, specifying 
the value of their associated robustness parameter and the number of local 
measurement settings required. For the detection of multipartite entanglement
in the neighbourhood of  the state $\ket{G_{n}}$, the
 projector-based witness reads

\begin{eqnarray}
	W_{n}=\frac{2^{n-1}-1}{2^{n-1}}\mathbb{I}-\ket{G_{n}}\bra{G_{n}},
\label{W-Gn}
\end{eqnarray}
which  needs $\frac{3^{n}-1}{2}$ local measurement settings, see App. B,
 and has robustness
parameter 
$p_{n}^{L}=\frac{2}{2^{n}-1}$. Because of the maximum-cardinality 
hyperedge, this 
case is the worst case scenario regarding the number of measurement 
settings. All of the following projective witnesses require a number of local 
measurement settings lower than or equal to (but possibly as high as) 
$\frac{3^{n}-1}{2}$. 

Starting from the state $\ket{H_{n}^{n-1}}$ we
have the projector-based witnesses

\begin{eqnarray}
	\begin{dcases}
		W_{4}=\frac{3+\sqrt{5}}{8}\mathbb{I}-\ket{H_{4}^{3}}
\bra{H_{4}^{3}},\,\,\,n=4\\
		W_{n}=\frac{2^{n-1}-n}{2^{n-1}}\mathbb{I}-\ket{H_{n}^{n-1}}
\bra{H_{n}^{n-1}},\,\,\,n\geq 6\,\text{even}\\
		W_{n}=\frac{2^{n-1}-n+1}{2^{n-1}}\mathbb{I}-
\ket{H_{n}^{n-1}}\bra{H_{n}^{n-1}},\,\,\,n\geq 3\,\text{odd}\,.
	\end{dcases}
\label{W-Hn-1}
\end{eqnarray}
with $p_{4}^{L}=\frac{10-2\sqrt{5}}{15}$, $p_{n}^{L}=\frac{2n}{2^{n}-1}$ for 
$n$ even, $p_{n}^{L}=\frac{2(n-1)}{2^{n}-1}$ for $n$ odd.

The projector-based witnesses for hypergraph states with hyperedges of
cardinality $n$ and $n-1$ read 

\begin{eqnarray}
	\begin{dcases}
		W_{3}=\frac{3}{4}\mathbb{I}-\ket{H_{3}^{2,3}}
\bra{H_{3}^{2,3}},\,\,\,n=3 \\
		W_{n}=\frac{2^{n-1}-n}{2^{n-1}}\mathbb{I}-
\ket{H_{n}^{n-1,n}}\bra{H_{n}^{n-1,n}},\,\,\,n\geq 5\,\text{odd}\\
		W_{n}=\frac{2^{n-1}-n+1}{2^{n-1}}\mathbb{I}-
\ket{H_{n}^{n-1,n}}\bra{H_{n}^{n-1,n}},\,\,\,n\geq 4\,\text{even}\,.
	\end{dcases}
\label{W-Hn,n-1}
\end{eqnarray}
with $p_{3}^{L}=\frac{2}{7}$, $p_{n}^{L}=\frac{2(n-1)}{2^{n}-1}$ for $n$ 
even, $p_{n}^{L}=\frac{2n}{2^{n}-1}$ for $n$ odd.

Finally, for a generic hypergraph state with hyperedges of maximum cardinality 
${k_{max}}$ we have the witness

\begin{eqnarray}
W_{n}=\frac{2^{k_{max}-1}-1}{2^{k_{max}-1}}\mathbb{I}-\ket{H^{k_{max}}_{n}}
\bra{H^{k_{max}}_{n}}\,
\label{W-Hn,n-1}
\end{eqnarray}
with robustness threshold $p_{n}^{L}=\frac{2^{n-k_{max}+1}}{2^{n}-1}$.

\subsection{Stabilizer entanglement witnesses}

We now construct entanglement witnesses of the form
\begin{equation}
\label{witness_stabilizer}
	\tilde{W}_{n}=\beta_{n} \mathbb{I}-\sum_{i=1}^{n}K_{i}
\end{equation}
with $\beta_{n}\in\mathbb{R}_{+}$, exploiting the stabilizer formalism and 
generalizing the procedure proposed in Refs.\hspace{-0.1cm} \cite{toth2005} 
and \cite{toth20052}. As mentioned above, projector-based entanglement 
witnesses need a number of 
local measurement settings that in general is exponentially growing with the 
number of qubits. 
The aim of the stabilizer construction is hence to improve this experimental 
efficiency. The stabilizer entanglement witnesses we propose indeed need a 
number of local measurement settings that grows linearly with the number of 
qubits. However, they are less fine that the projector-based ones and 
display a lower robustness parameter.\\

In order to determine suitable values 
for $\beta_n$ such that $\tilde{W}_{n}$ is an entanglement witness  we require 
that $\tilde{W}_{n} - C\,W_{n}\geq 0$ for some positive constant $C>0$. If 
this holds we have that
\begin{equation}
\label{witnesses_inequality}
	Tr[\rho \tilde{W}_{n}]\geq CTr[\rho W_{n}]
\end{equation}
and $\tilde{W}_{n}$ is still a good entanglement witness. Its robustness 
parameter is $\frac{n-\beta_{n}}{n}$. In order to maximize 
$\tilde{p}_{L}^{n}$ we hence need to minimize $\beta_{n}$.\\

In order to require (\ref{witnesses_inequality}) we compare the two witnesses 
$W_{n}$ and $\tilde{W}_{n}$ by means of the hypergraph state basis. 
The action of $W_{n}$ on the hypergraph state basis \eqref{def-Hyperbasis} is

\begin{eqnarray}
	\begin{cases}
		W_{n}\ket{\phi_{00...0}}=\alpha(\ket{H_{n}}) -1\\
		W_{n}\ket{\phi_{x\neq 00...0}}=\alpha(\ket{H_{n}})
	\end{cases}
\label{W-basis}
\end{eqnarray}
Inequality \eqref{witnesses_inequality} results in the following set of 
constraints for the parameter $\beta_{n}$
\begin{eqnarray}
	\begin{cases}
		0>\beta_{n}-n \geq C\left(\alpha(\ket{H_{n}})-1\right)\\
		\beta_{n}+n \geq C\alpha(\ket{H_{n}})\\
		\beta_{n} \pm (n-2) \geq C\alpha(\ket{H_{n}})\\
		...\\
		\beta_{n} \pm (n-2m) \geq C\alpha(\ket{H_{n}})\\
		...\\
		\beta_{n}\geq C\alpha(\ket{H_{n}})
	\end{cases}
\label{const-neven}
\end{eqnarray}
for $n$ even, $m$ positive integer with $0<m<n/2$ and $C>0$, while
\begin{eqnarray}
	\begin{cases}
	0>\beta_{n}-n \geq C\left(\alpha(\ket{H_{n}})-1\right)\\
		\beta_{n}+n \geq C\alpha(\ket{H_{n}})\\
		\beta_{n} \pm (n-2) \geq C\alpha(\ket{H_{n}})\\
		...\\
		\beta_{n} \pm (n-2m) \geq C\alpha(\ket{H_{n}})\\
		...\\
		\beta_{n} \pm 1 \geq C\alpha(\ket{H_{n}})
	\end{cases}
\label{const-nodd}
\end{eqnarray}
for $n$ odd, $m$ and $C$ as above.\\
These inequalities are all compatible with each other for every $n$ and 
define a convex compatibility region for $\beta_{n}$ and $C$ where one can 
minimize $\beta_{n}$ (see Figure \ref{fig_feasregion} for a pictorial 
representation). The first two inequalities (first line) in the set of 
constraints are identified by the region below the black line and the one on 
the right of the blue one, while the inequality in the third line of the set is
identified by the region on the left of the green line. Notice that if the 
latter is satisfied then also all the remaining constraints in the set are 
automatically satisfied. 
The optimal value for $\beta_{n}$, given by the minimum compatible with the
set of constraints, is at the intersection of the (blue and green) lines
\begin{eqnarray}
	\begin{cases}
	\beta_{n}=n -C\left(1-\alpha(\ket{H_{n}})\right)\\
	\beta_{n}=(n-2)+ C\alpha(\ket{H_{n}})
	\end{cases}
\label{2lines}
\end{eqnarray}
which corresponds to $C=2$ and $\beta_{n}=n-2(1-\alpha(\ket{H_{n}})$.

\begin{figure}[h!]
\centering
	\begin{tikzpicture}[scale=0.8]
		\shade[bottom color=gray,top color=white] (0,4)-- (1,4)--(2/3,10/3);
		\path (0,6) node[] (xO) {\small $\beta_{n}$};
		\path (0,0) node[] (O) {};
		\path (6,0) node[] (yO) {\small $C\cdot\left(1-\alpha(\ket{H_{n}})\right)$};
		\draw[->,thick,black] (O.center)--(xO);
		\draw[->,thick,black] (O.center)--(yO);
		\path (0,2) node[] (n2) {$\bullet$};
		\path (0,4) node[] (n) {$\bullet$};
		\path (6,4) node[] (nline) {};
		\draw[black] (n.center)--(nline);
		\path (-0.8,2) node[] (xlabel1) {\small $n-2$};
		\path (-0.5,4) node[] (xlabel2) {\small $n$};
		\path (4,0) node[] (line1) {};
		\path (2,6) node[] (line2) {};
		\draw[thick,blue,dashed] (n.center)--(line1.center);
		\draw[thick,green,densely dotted] (n2.center)--(line2.center);
		\path (0,2) node[] () {$\bullet$};
		\path (0,4) node[] () {$\bullet$};
		\path (4,6) node[] (label1) {};
		\path (4,2) node[] (label2) {};
		\path[red] (2/3,10/3) node[] (sol) {$\bullet$};
	\end{tikzpicture}
	\caption{Grey region: feasible region. Black line: $\beta_{n}=n$. Blue dashed line: $\beta_{n}=n -C\left(1-\alpha(\ket{H_{n}})\right)$. Green dotted line: $\beta_{n}=(n-2)+ C\alpha(\ket{H_{n}})$. Red dot: optimal value for $\beta_{n}$.}
	\label{fig_feasregion}
\end{figure}
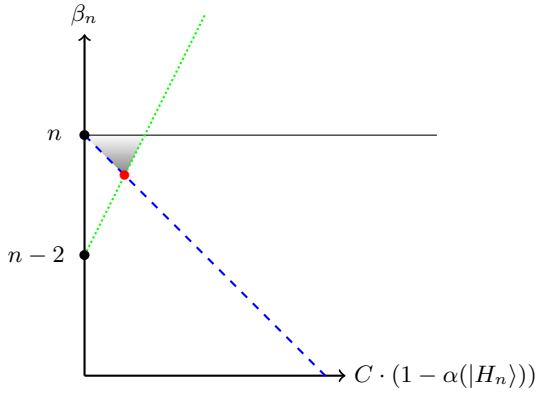

In the following we list the stabilizer entanglement witnesses constructed
in ths way, 
specifying the value of their associated robustness parameter and the number 
of local measurement settings required.\\

A suitable witness for the detection of entanglement in the neighbourhood of 
$\ket{G_{n}}$ is given by

\begin{eqnarray}
	\tilde{W}_{n}=\frac{n2^{n-1}-2}{2^{n-1}}\mathbb{I}-\sum_{i=1}^{n}K_{i}
(G_n)
\label{tildeG}
\end{eqnarray}
requiring exactly $n$ local measurement settings and with 
$\tilde{p}_{n}^{L}=\frac{2}{n2^{n-1}}$. 
Here $K_{i}(G_n)$ denotes the stabilizer operators for hypergraph
state $\ket{G_{n}}$.
Just like in the projective case, 
this is the worst case scenario  regarding the number of measurement 
settings. This means that the number of measurements required by the stabilizer 
witnesses that follow is limited from above by $n$, hence no more by an 
exponential function of the number of qubits but by a linear one.
A comparison between the robustness parameter for the projective witness 
(\ref{W-Gn}) and the stabilizer witness (\ref{tildeG}) is reported in Fig. 
\ref{fig_robustness}. As we can see, stabilizer entanglement witnesses need 
fewer local measurement settings than the projective ones but, as a drawback, 
they are less robust to noise than the projective ones.
\\
Suitable stabilizer witnesses for the detection of entanglement in the 
neighbourhood of $\ket{H_{n}^{n-1}}$ are

\begin{eqnarray}
	\begin{dcases}
		\tilde{W_{4}}=\frac{11+\sqrt{5}}{4}\mathbb{I}
-\sum_{i=1}^{4}K_{i}(H_{4}^{3}),\,\,\,n=4\\
		\tilde{W}_{n}=\frac{n2^{n-1}-2n}{2^{n-1}}\mathbb{I}-
\sum_{i=1}^{n}K_{i}(H_{n}^{n-1}),\,\,\,n\geq 6\,\text{even}\\
		\tilde{W}_{n}=\frac{n2^{n-1}-2n+2}{2^{n-1}}\mathbb{I}
-\sum_{i=1}^{n}K_{i}(H_{n}^{n-1}),\,\,\,n\geq 3\,\text{odd}\,.
	\end{dcases}
\label{tildeWn}
\end{eqnarray}
with $\tilde{p}_{4}^{L}=\frac{5-\sqrt{5}}{16}$, 
$\tilde{p}_{n}^{L}=\frac{2}{2^{n-1}}$ for $n$ even, 
$\tilde{p}_{n}^{L}=\frac{2(n-1)}{n2^{n-1}}$ for $n$ odd.\\
For the detection of entanglement in the neighbourhood of 
$\ket{H_{n}^{n-1,n}}$ we have

\begin{eqnarray}
	\begin{dcases}
		\tilde{W}_{3}=\frac{5}{2}\mathbb{I}
-\sum_{i=1}^{3}K_{i}(H_{3}^{2,3}),\,\,\,n=3\\
		\tilde{W}_{n}=\frac{n2^{n-1}-2n}{2^{n-1}}\mathbb{I}
-\sum_{i=1}^{n}K_{i}(H_{n}^{n-1,n}),\,\,\,n\geq 5\,\text{odd}\\
		\tilde{W}_{n}=\frac{n2^{n-1}-2n+2}{2^{n-1}}\mathbb{I}
-\sum_{i=1}^{n}K_{i}(H_{n}^{n-1,n}),\,\,\,n\geq 4\,\text{even}\,.
	\end{dcases}
\label{tildeWn,n-1}
\end{eqnarray}
with $\tilde{p}_{3}^{L}=\frac{1}{6}$, $\tilde{p}_{n}^{L}=\frac{2(n-1)}
{n2^{n-1}}$ for $n$ even, $\tilde{p}_{n}^{L}=\frac{2}{2^{n-1}}$ for $n$ odd.\\
Suitable stabilizer witnesses for the detection of entanglement in the 
neighbourhood of a generic hypergraph state with hyperedges with
maximum cardinality ${k}_{max}$ are given by

\begin{eqnarray}
	\tilde{W}_{n}=\frac{n2^{k_{max}-1}-2}{2^{k_{max}-1}}
\mathbb{I}-\sum_{i=1}^{n}K_{i}(H_{k_{max}})\,
\label{tildeWgen}
\end{eqnarray}
with $\tilde{p}_{n}^{L}=\frac{2}{n2^{k_{max}-1}}$.
A derivation of bounds on the number of measurement settings required to
measure the expectation values of the projective witnesses versus the 
stabilizer ones is reported in Appendix \ref{appendix_measurement}.

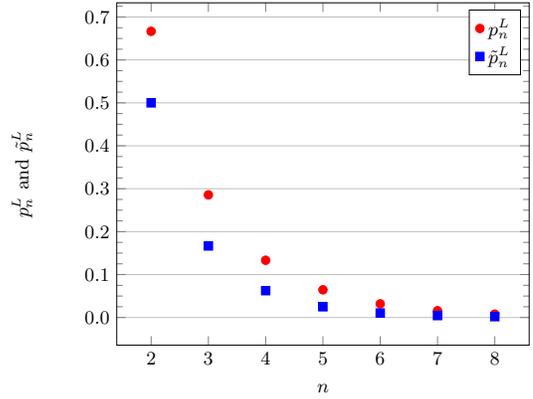
\begin{figure}[h!]
	\centering
	\begin{tikzpicture}[scale=0.8]
		\begin{axis}[xlabel=$n$,ylabel=$p_{n}^{L}$ {\small and} $\tilde{p}_{n}^{L}$,tick label style={/pgf/number format/.cd,fixed, precision=2},ytick={0.0,0.1,0.2,0.3,0.4,0.5,0.6,0.7},ylabel style={at={(-0.05,0.5)}},ymajorgrids,minor y tick num={3},minor grid style={draw=none},
		y tick label style={/pgf/number format/.cd,fixed,fixed zerofill,precision=1}]
			\addplot [only marks,color=red,mark=*] coordinates {
			(2,2/3)
			(3,2/7)
			(4,2/15)
			(5,2/31)
			(6,2/63)
			(7,2/127)
			(8,2/255)};
			\addlegendentry{$p_{n}^{L}$};
			\addplot [only marks,color=blue,mark=square*] coordinates {
			(2,1/2)
			(3,1/6)
			(4,1/16)
			(5,1/40)
			(6,1/96)
			(7,1/224)
			(8,1/512)};
			\addlegendentry{$\tilde{p}_{n}^{L}$};
		\end{axis}
	\end{tikzpicture}
	\caption{Entanglement detection in the neighbourhood of a $n$-qubit hypergraph state $\ket{G_{n}}$ with only one maximum-cardinality hyperedge. Plot of the robustness parameters $p^{L}_{n}$ (red line with dots) and $\tilde{p}^{L}_{n}$ (blue line with squares) versus the number of qubits $n$ for $2\leq n\leq 8$. }
	\label{fig_robustness}
\end{figure}

\section{Conclusions}
\label{Concl}

In this work we have studied the entanglement properties of some classes of
symmetric hypergraph states for an arbitrary number of qubits $n$ by 
proposing an analytical procedure based on the notion of infinity norm
of a matrix. We have also derived lower bounds to the multipartite entanglement
of a generic hypergraph state, that depends on the value of the maximal 
cardinality of the corresponding hypergraph. We have then constructed two 
classes of entanglement witness operators for the detection of multipartite
entanglement in the neighborhood of a hypergraph state and we have compared
their efficiency in terms of minimal number of measurement settings required
in order to measure the  expectation value of the witness operator and in terms
of the corresponding robustness parameter.
The analysis has shown that projector-based witnesses perform better in terms of
robustness parameter with respect to stabilizer witnesses, but in general they 
are more demanding in terms of the number of measurement settings required.

\begin{acknowledgments}
We thank O. G\"uhne and M. Gachechiladze for suggestions to improve a
previous version of this manuscript.\\
DB acknowledges support by BMBF.
\end{acknowledgments}

\appendix
\section{Proofs of the multipartite entanglement formulas}
\label{appendix_multipartite_entanglement}

It suffices here to prove the overlap and multipartite entanglement formulas 
of Theorem \ref{theorem_multipartite_entanglement_n-1}. The proof of Theorem 
\ref{theorem_multipartite_entanglement_n-1n} is based on the same procedure, 
with the only difference being in the role played by the parity of the number 
of qubits and the exception of case $n=3$.
\begin{proof}
\textit{Case $n$ even.} We consider the representation of hypergraph state 
$\ket{H_{n}^{n-1}}$ over the computational basis and note that the number of 
negative signs in front of the computational basis elements is equal to 
$\binom{n}{n-1}\equiv n$, while the coefficient of state $\ket{11...\,1}$ is 
equal to $+1$. The Schmidt decomposition of hypergraph state 
$\ket{H_{n}^{n-1}}$ with respect to bipartition $\bar A=\lbrace 1,2,...\,(n-1) 
\rbrace$ and $\bar B=\lbrace n \rbrace$ may always be written as
\begin{equation}
\label{even_hypergraph_decomposition}
	\begin{split}
		\ket{H_{n}^{n-1}}&=\sqrt{\frac{2^{n-1}-n}{2^{n-1}}}
\sum_{\underset{w(x)<n-2}{x=0,}}^{2^{n-1}-1}\frac{\ket{x}}{\sqrt{2^{n-1}-n}}
\ket{+}\,+\\
		&+\sqrt{\frac{n}{2^{n-1}}}\cdot\frac{1}{\sqrt{n}}
\left(\sum_{\underset{w(x)=n-2}{x=0,}}^{2^{n-1}-1}\ket{x}-\ket{11...\,1} 
\right)\ket{-}
	\end{split}
\end{equation}
where the first summation is taken over all binary numbers from $0$ to 
$2^{n-1}-1$ with weight $w(x)<n-2$ while the second summation is taken over 
those binary numbers $x$ that have weight $w(x)$ exactly equal to $n-2$; 
while the first summation is made up of $2^{n-1}-n$ addends, the second one 
is made up of $n-1$ terms.\\
As a straightforward consequence of representation 
\eqref{even_hypergraph_decomposition}, it follows that the maximum Schmidt 
coefficient of hypergraph state $\ket{H_{n}^{n-1}}$, with respect to 
bipartition $\bar A$ $\bar B$ and for $n\geq 4$, is 
$s^{\bar A \bar B}_{max}(\ket{H_{n}^{n-1}})=\sqrt{\frac{2^{n-1}-n}{2^{n-1}}}$.
\paragraph*{Case $n=4$.} We first consider the case with $n=4$ and write 
hypergraph state $\ket{H_{4}^{3}}$ as
\begin{equation}
	\begin{split}
		\ket{H_{4}^{3}}&=\frac{1}{\sqrt{2}}
\Bigl( \frac{\ket{000}+\ket{001}+\ket{010}+\ket{100}}{2}\ket{+}\,+\\
		&+ \frac{\ket{011}+\ket{101}+\ket{110}-
\ket{111}}{2}\ket{-}\Bigr)\;.
	\end{split}
\end{equation}
The maximum Schmidt coefficient with respect to bipartition 
$A=\lbrace 1,2,3\rbrace$ and $B=\lbrace 4 \rbrace$ is 
$s^{AB}_{max}(\ket{H_{4}^{3}}) = \frac{1}{\sqrt{2}}$, hence the maximum 
eigenvalue of the corresponding reduced density matrix $\rho^{(123)}$ is 
$\lambda_{max}(\rho^{(123)})=\frac{1}{2}$.\\
Taking then into consideration bipartition $A = \lbrace 1,2 \rbrace$ and 
$B = \lbrace 3,4 \rbrace$, we write hypergraph state $\ket{H_{4}^{3}}$ as

\begin{eqnarray}
	\begin{split}
		\ket{H_{4}^{3}}&=\frac{\ket{0}\ket{G_{3}}
+\ket{1}C_{3}^{(2,3,4)}\ket{H_{3}^{2}}}{\sqrt{2}}\\
		&=\frac{1}{2}\left(\ket{00}\ket{+}^{\otimes 2}
+(\ket{01}+\ket{10})\ket{G_{2}}+\ket{11}\ket{H_{2}^{1}}  \right)\,.
	\end{split}
\end{eqnarray}
If we now partially trace over both systems $3$ and $4$, the resulting 
density matrix is given by
\begin{eqnarray}
	\begin{split}
		\rho^{(12)}&=\frac{1}{16}
		\left( \begin{matrix}
			4 & 2 & 2 & 0 \\
			2 & 4 & 4 & -2\\
			2 & 4 & 4 & -2\\
			0 & -2 & -2 & 4
		\end{matrix} \right)
	\end{split}
\end{eqnarray}
hence $\lambda_{max}(\rho^{(12)})\leq\parallel \hspace{-0.1cm}
\rho^{(12)}\hspace{-0.1cm}\parallel_{\infty}=\frac{3}{4}$. 
We conclude that $\alpha (\ket{H_{4}^{3}})\leq \frac{3}{4}$.\\
If we explicitly compute the eigenvalues of the matrix $\rho^{(12)}$ it
turns out that its maximum eigenvalue is $\lambda_{max}(\rho^{(12)})
=\frac{3+\sqrt{5}}{8}\leq\frac{3}{4}$, as expected.
\paragraph*{Case $n\geq 6$.} Let $\ket{H_{n}^{n-1}}$ be a $n$-qubit 
hypergraph state with $n$ even, $n\geq 6$, and with all possible 
$(n-1)$-hyperedges. Besides the representation 
\eqref{even_hypergraph_decomposition}, the hypergraph state 
$\ket{H_{n}^{n-1}}$ may be written in the following ways
\begin{equation}
	\begin{split}
		\ket{H_{n}^{n-1}}&=
\frac{\ket{0}\ket{G_{n-1}}+\ket{1}\ket{\tilde{H}_{n-1}^{n-2}}}{\sqrt{2}}\\
		&=\frac{\ket{00}\ket{+}^{\otimes n-2}+(\ket{01}+\ket{10})
\ket{G_{n-2}}+\ket{11}\ket{\tilde{H}_{n-2}^{n-3}}}{2}\\
		&=\frac{1}{2\sqrt{2}}\Bigl( (\ket{000}+\ket{001}+\ket{010}
+\ket{100})\ket{+}^{\otimes n-3}+\\
		&+(\ket{011}+\ket{101}+\ket{110})\ket{G_{n-3}}+\ket{111}
\ket{\tilde{H}_{n-3}^{n-4}} \Bigr)\\
		&=\frac{1}{\sqrt{2^{k}}}\Biggl(\sum_{\underset{w(x)<k-1}
{x=0,}}^{2^{k}-1}\ket{x}\ket{+}^{\otimes n-k}+\\
		&+\sum_{\underset{w(x)=k-1}{x=0,}}^{2^{k}-1}\ket{x}
\ket{G_{n-k}}+\ket{\overbrace{11...\,1}^{k}}\ket{\tilde{H}_{n-k}^{n-k-1}}
\Biggr)
	\end{split}
\end{equation}
for all $1\leq k \leq n$, where $\ket{\tilde{H}_{n-k}^{n-k-1}}$ is the 
$(n-k)$-qubit hypergraph state with all possible $(n-k-1)$-hyperedges plus an 
additional $n$-hyperedge if $n-k$ is odd, so that the coefficient in front 
of $\ket{11...\,1}$ remains equal to $+1$ for every $k$. While the first 
summation is made up of $2^{k}-k-1$ terms, the second summation is made up 
of $k$ addends. As a consequence, we may express the corresponding reduced 
density matrix $\rho^{(n-k)}$ as
\begin{equation}
	\rho^{(n-k)}=\frac{1}{\mathcal{N}}(2^{k}-k-1)\mathcal{I}_{2^{n-k}}
+\frac{1}{\mathcal{N}}k\tilde{\mathcal{I}}_{2^{n-k}}+\frac{1}{\mathcal{N}}
\mathcal{R}_{2^{n-k}}
\end{equation}
where $\mathcal{I}_{2^{n-k}}\propto (\ket{+}\bra{+})^{\otimes n-k}$ is the 
$2^{n-k}\times 2^{n-k}$ matrix filled with ones, 
$\tilde{\mathcal{I}}_{2^{n-k}}\propto\ket{G_{n-k}}\bra{G_{n-k}}$ is the 
$2^{n-k}\times 2^{n-k}$ matrix filled with ones except for the last column 
and for the last row, that have all elements equal to $-1$ and a final $1$, 
i.e.
\begin{displaymath}
	\tilde{\mathcal{I}}_{2^{n-k}}=\left( \begin{matrix}
		1 & 1 & \dots & 1 & -1 \\
		1 & 1 & \dots & 1 & -1 \\
		\vdots & \vdots & \ddots & \vdots & \vdots \\
		1 & 1 & \dots & 1 & -1 \\
		-1 & -1 & \dots & -1 & 1
	\end{matrix} \right)
\end{displaymath}
$\mathcal{R}_{2^{n-k}}\propto \ket{\tilde{H}_{n-k}^{n-k-1}}
\bra{\tilde{H}_{n-k}^{n-k-1}}$ and $\mathcal{N}$ is a normalization factor. 
The matrix $\mathcal{R}_{2^{n-k}}$ has two kinds of rows: the first one has 
$n-k$ elements equal to $-1$ and $2^{n-k}-n+k$ elements equal to $1$, in 
particular the last element is equal to $1$; the second one has the same 
elements with opposite sign. The sum of these three matrices gives a matrix 
whose elements may assume the following $4$ possible values
\begin{displaymath}
	\begin{split}
		&v_{1}=2^{k}-k-1+k+1=2^{k}\\
		&v_{2}=2^{k}-k-1+k-1=2^{k}-2\\
		&v_{3}=2^{k}-k-1-k+1=2^{k}-2k\\
		&v_{4}=2^{k}-k-1-k-1=2^{k}-2k-2
	\end{split}
\end{displaymath}
where $v_{1}$ and $v_{2}$ are internal values while $v_{3}$ and $v_{4}$ are 
values that may be assumed by elements of the last column and row. Values 
$v_{i}$ for $i=1,2,3,4$ are all non-negative except for $v_{4}=-2$ when 
$k=2$.\\
We can now evaluate $\parallel \rho^{n-k} \parallel_{\infty}$ as
\begin{equation}
	\parallel \rho^{n-k} \parallel_{\infty}=\max \lbrace 
\mathcal{N}^{1}_{\infty},\mathcal{N}^{2}_{\infty}\rbrace
\end{equation}
where $\mathcal{N}^{1}_{\infty}$ gathers contributions coming from the first 
row while $\mathcal{N}^{2}_{\infty}$ has contributions coming from the other 
possible kind of row.\\
We will first consider $k=2$. In this case it turns out 
that $v_{1}=4$, $v_{2}=2$, $v_{3}=0$ and $v_{4}=-2$. This means that 
$\mathcal{N}^{1}_{\infty}$ and $\mathcal{N}^{2}_{\infty}$ take the following 
values
\begin{displaymath}
	\begin{split}
		\mathcal{N}^{1}_{\infty}
		\overset{k=2}{=}\frac{2^{n-1}-n}{2^{n-1}}
	\end{split}
\end{displaymath}
\begin{displaymath}
	\begin{split}
		\mathcal{N}^{2}_{\infty}
		\overset{k=2}{=}\frac{2^{n-1}-n}{2^{n-1}}
-\left( \frac{2^{n-2}-2n+2}{2^{n-1}}\right)
	\end{split}
\end{displaymath}
where $\frac{2^{n-2}-2n+2}{2^{n-1}}\geq 0$ for $n\geq 5$. This means that, 
when $k=2$, $\mathcal{N}^{1}_{\infty}\geq \mathcal{N}^{2}_{\infty}$ for 
every $n\geq 6$, with the only exception of the case $n=4$ that has already 
been examined above. Actually when $n=4$ and $k=2$, 
we have that $\frac{2^{n-1}-n}{2^{n-1}}=\frac{1}{2}$ and 
$\frac{2^{n-2}-2n+2}{2^{n-1}}=-\frac{1}{4}$, so the sum of their absolute 
values is $\frac{3}{4}$ as expected.\\
Consider now $k\geq 3$. In this case we have that 
$\mathcal{N}^{1}_{\infty}$ and $\mathcal{N}^{2}_{\infty}$ may be written as
\begin{displaymath}
	\begin{split}
		\mathcal{N}^{1}_{\infty}&=\frac{v_{1}(2^{n-k}-n+k-1)+v_{2}
(n-k)+v_{3}}{\mathcal{N}}\\
		&=\frac{2^{n-1}-n}{2^{n-1}}
	\end{split}
\end{displaymath}
\begin{displaymath}
	\begin{split}
		\mathcal{N}^{2}_{\infty}&=\frac{v_{2}(2^{n-k}-n+k-1)
+v_{1}(n-k)+v_{4}}{\mathcal{N}}\\
		&=\frac{2^{n-1}-n}{2^{n-1}}-\left( \frac{2^{n-k}-2(n-k)}
{2^{n-1}} \right)
	\end{split}
\end{displaymath}

where $\frac{2^{n-k}-2(n-k)}{2^{n-1}}\geq 0$ for $n-k\geq 2$. 
Since we need to consider only inequivalent bipartitions, that means 
$k\leq \frac{n}{2}$, the cases corresponding to $n\geq 6$ and 
$3\leq k\leq \frac{n}{2}$ satisfy this requirement. Hence we have that 
$\mathcal{N}^{1}_{\infty}\geq \mathcal{N}^{2}_{\infty}$ for all values of 
$n\geq 6$ and $3\leq k\leq \frac{n}{2}$.\\
We conclude by observing that the value of $\mathcal{N}^{1}_{\infty}$ is 
independent of $k$ and equal to the squared maximum Schmidt coefficient 
$s_{max}^{\bar A \bar B}(\ket{H^{n}_{n-1}})^{2}$, evaluated with respect to 
bipartition $\bar A \bar B$.\\

\textit{Case $n$ odd.} We consider the representation of hypergraph state 
$\ket{H_{n}^{n-1}}$ over the computational basis and note that the number of 
negative signs in front of the computational basis elements is equal to 
$\binom{n}{n-1}+1\equiv n+1$, in particular the coefficient of state 
$\ket{11...\,1}$ is equal to $-1$. The Schmidt decomposition of the 
hypergraph state $\ket{H_{n}^{n-1}}$ with respect to bipartition 
$\bar A \bar B$ may always be written as
\begin{equation}
\label{odd_hypergraph_decomposition}
	\begin{split}
		\ket{H_{n}^{n-1}}&=\sqrt{\frac{2^{n-1}-n+1}{2^{n-1}}}\cdot 
\frac{1}{\sqrt{2^{n-1}-n+1}}\cdot\\
		&\cdot\left(\sum_{\underset{w(x)<n-2}{x=0,}}^{2^{n-1}-1}
\ket{x}-\ket{11...\,1}\right)\ket{+}\,+\\
		&+\sqrt{\frac{n-1}{2^{n-1}}}\sum_{\underset{w(x)=n-2}
{x=0,}}^{2^{n-1}-1}\frac{\ket{x}}{\sqrt{n-1}}\ket{-}
	\end{split}
\end{equation}
where the first summation is taken over all binary numbers ranging from $0$ 
to $2^{n-1}-1$ with weight $w(x)<n-2$ while the second summation is taken 
over those binary numbers $x$ that have weight $w(x)$ exactly equal to $n-2$; 
while the first summation is made up of $2^{n-1}-n$ addends, the second one 
is made up of $n-1$ terms.\\
As a straightforward consequence of representation 
\eqref{odd_hypergraph_decomposition}, it follows that the maximum Schmidt 
coefficient of hypergraph state $\ket{H_{n}^{n-1}}$, with respect to 
bipartition $\bar A \bar B$, is 
$s^{\bar A \bar B}_{max}(\ket{H_{n}^{n-1}})
=\sqrt{\frac{2^{n-1}-n+1}{2^{n-1}}}$.
\paragraph*{Case $n=3$.} We first consider case $n=3$ and write the 
hypergraph state $\ket{H_{3}^{2}}$ as
\begin{equation}
\ket{H_{3}^{2}}=\frac{1}{\sqrt{2}}\frac{\ket{00}-\ket{11}}{\sqrt{2}}
\ket{+}+\frac{1}{\sqrt{2}}\frac{\ket{01}+\ket{10}}{\sqrt{2}}\ket{-}\,.
\end{equation}
The maximum Schmidt coefficient with respect to bipartition 
$A=\lbrace 1,2\rbrace$ and $B=\lbrace 3 \rbrace$ is 
$s^{AB}_{max}(\ket{H_{3}^{2}})=\frac{1}{\sqrt{2}}$, hence the maximum 
eigenvalue of the corresponding reduced density matrix $\rho^{(12)}$ is 
$\lambda_{max}(\rho^{(12)})=\frac{1}{2}$ and the maximum overlap is 
$\alpha(\ket{H_{3}^{2}})=\frac{1}{2}$. We do not need to consider any other 
bipartitions since, given the invariance of the state under permutations of 
the qubits, they are all equivalent to this one.
\paragraph*{Case $n\geq 5$.} Let $\ket{H_{n}^{n-1}}$ be a $n$-qubit 
hypergraph state with $n$ odd, $n\geq 5$, and with all possible 
$(n-1)$-hyperedges. Besides the representation 
\eqref{odd_hypergraph_decomposition}, the hypergraph state 
$\ket{H_{n}^{n-1}}$ may be written as
\begin{equation}
	\begin{split}
\ket{H_{n}^{n-1}}&=\frac{\ket{0}\ket{G_{n-1}}+\ket{1}
\ket{\hat{H}_{n-1}^{n-2}}}{\sqrt{2}}\\
&=\frac{\ket{00}\ket{+}^{\otimes n-2}+(\ket{01}+\ket{10})
\ket{G_{n-2}}+\ket{11}\ket{\hat{H}_{n-2}^{n-3}}}{2}\\
&=\frac{1}{2\sqrt{2}}\Bigl( (\ket{000}+\ket{001}+\ket{010}
+\ket{100})\ket{+}^{\otimes n-3}+\\
&+(\ket{011}+\ket{101}+\ket{110})\ket{G_{n-3}}+\ket{111}
\ket{\hat{H}_{n-3}^{n-4}} \Bigr)\\
&=\frac{1}{\sqrt{2^{k}}}\Biggl(\sum_{\underset{w(x)<k-1}
{x=0,}}^{2^{k}-1}\ket{x}\ket{+}^{\otimes n-k}+\sum_{\underset{w(x)=k-1}
{x=0,}}^{2^{k}-1}\ket{x}\ket{G_{n-k}}+\\
		&+\ket{\overbrace{11...\,1}^{k}}\ket{\hat{H}_{n-k}^{n-k-1}}
\Biggr)
	\end{split}
\end{equation}
for all $1\leq k \leq n$, where $\ket{\hat{H}_{n-k}^{n-k-1}}$ is the 
$(n-k)$-qubit hypergraph state with all possible $(n-k-1)$-hyperedges plus 
an additional $n$-hyperedge if $n-k$ is even, so that the coefficient in 
front of $\ket{11...\,1}$ remains equal to $-1$ for every $k$. While the 
first summation is made up of $2^{k}-k-1$ terms, the second summation is 
made up of $k$ addends. As a consequence, we may represent the corresponding 
reduced density matrix $\rho^{(n-k)}$ as
\begin{equation}
	\rho^{(n-k)}=\frac{1}{\mathcal{N}}(2^{k}-k-1)\mathcal{I}_{2^{n-k}}
+\frac{1}{\mathcal{N}}k\tilde{\mathcal{I}}_{2^{n-k}}+\frac{1}{\mathcal{N}}
\mathcal{R}_{2^{n-k}}
\end{equation}
where $\mathcal{R}_{2^{n-k}}\propto \ket{\hat{H}_{n-k}^{n-k-1}}
\bra{\hat{H}_{n-k}^{n-k-1}}$ and $\mathcal{N}$ is a normalization factor. The
matrix $\mathcal{R}_{2^{n-k}}$ has two kinds of rows: the first one has 
$n-k+1$ elements equal to $-1$ and $2^{n-k}-n+k-1$ elements equal to $1$, 
in particular the last element is equal to $-1$; the second one has the same 
elements with opposite sign. The sum of these three matrices gives a matrix 
whose elements may assume the following $4$ possible values
\begin{displaymath}
	\begin{split}
		&v_{1}=2^{k}-k-1+k+1=2^{k}\\
		&v_{2}=2^{k}-k-1+k-1=2^{k}-2\\
		&v_{3}=2^{k}-k-1-k-1=2^{k}-2k-2\\
		&v_{4}=2^{k}-k-1-k+1=2^{k}-2k
	\end{split}
\end{displaymath}
where $v_{1}$ and $v_{2}$ are internal values while $v_{3}$ and $v_{4}$ are 
values that may be assumed by elements of the last column and row. 
Values $v_{i}$ for $i=1,2,3,4$ are all non-negative except for $v_{3}=-2$ 
when $k=2$.\\
We can now evaluate $\parallel \hspace{-1mm}\rho^{n-k} 
\hspace{-1mm}\parallel_{\infty}$ as
\begin{equation}
	\parallel \hspace{-1mm}\rho^{n-k} \hspace{-1mm}\parallel_{\infty}=
\max \lbrace \mathcal{N}^{1}_{\infty},\mathcal{N}^{2}_{\infty}\rbrace
\end{equation}
where $\mathcal{N}^{1}_{\infty}$ gathers contributions coming from the first 
row while $\mathcal{N}^{2}_{\infty}$ has contributions coming from the other 
possible kind of row.\\
We consider first $k=2$. In this case it turns out 
that $v_{1}=4$, $v_{2}=2$, $v_{3}=-2$ and $v_{4}=0$. This means that 
$\mathcal{N}^{1}_{\infty}$ and $\mathcal{N}^{2}_{\infty}$ take the following 
values
\begin{displaymath}
	\begin{split}
		\mathcal{N}^{1}_{\infty}
		\overset{k=2}{=}\frac{2^{n-1}-n+1}{2^{n-1}}
	\end{split}
\end{displaymath}
\begin{displaymath}
	\begin{split}
		\mathcal{N}^{2}_{\infty}
		\overset{k=2}{=}\frac{2^{n-1}-n+1}{2^{n-1}} - 
\left(\frac{2^{n-3}-n+2}{2^{n-2}} \right)
	\end{split}
\end{displaymath}
where $\frac{2^{n-3}-n+2}{2^{n-2}} \geq 0$ for $n\geq 3$. This means that, 
when $k=2$, we have that $\mathcal{N}^{1}_{\infty}\geq 
\mathcal{N}^{2}_{\infty}$ for every $n\geq 5$.\\
Consider now $k\geq 3$. In this case we need to consider 
values of $n$ such that $n\geq 7$. We then find that 
$\mathcal{N}^{1}_{\infty}$ and $\mathcal{N}^{2}_{\infty}$ may be written as
\begin{displaymath}
	\begin{split}
		\mathcal{N}^{1}_{\infty}&=\frac{v_{1}(2^{n-k}-n+k-1)
+v_{2}(n-k)+v_{3}}{\mathcal{N}}\\
		&=\frac{2^{n-1}-n-1}{2^{n-1}}\\
		&\leq \frac{2^{n-1}-n+1}{2^{n-1}}
	\end{split}
\end{displaymath}
\begin{displaymath}
	\begin{split}
		\mathcal{N}^{2}_{\infty}&=\frac{v_{2}(2^{n-k}-n+k-1)
+v_{1}(n-k)+v_{4}}{\mathcal{N}}\\
		&=\frac{2^{n-1}-n-1}{2^{n-1}}-\left( 
\frac{2^{n-k}-2(n-k)-2}{2^{n-1}} \right)
	\end{split}
\end{displaymath}
where $\frac{2^{n-k}-2(n-k)}{2^{n-1}}\geq 0$ for $n-k\geq 3$. Since we need 
to consider only inequivalent bipartitions, that means $k\leq \frac{n-1}{2}$, 
the cases corresponding to $n\geq 7$ and $3\leq k\leq \frac{n-1}{2}$ satisfy 
this requirement. Hence we have that $\mathcal{N}^{1}_{\infty}
\geq \mathcal{N}^{2}_{\infty}$ also for all values of $n\geq 7$ and 
$3\leq k\leq \frac{n-1}{2}$.\\
We conclude by observing that the value of $\mathcal{N}^{1}_{\infty}$ is 
independent of $k$ and it is lower than the squared maximum Schmidt 
coefficient $s_{max}^{\bar A \bar B}(\ket{H_{n}^{n-1}})^{2}$, evaluated with 
respect to bipartition $\bar A\bar B$.
\end{proof}


\section{Measurement of the witnesses}
\label{appendix_measurement}
\paragraph{Single maximum-cardinality hyperedge case} In order to measure each 
stabilizer operator, one local measurement setting is required and they are 
all different from each other. To be more precise, stabilizer $K_{i}$ 
is given by the tensor product of the Pauli matrix $\pauli{x}^{(i)}$ and 
the control gate $C_{n-1}^{(1,2,...\hat{i},...\,n)}$, whose decomposition 
over the Pauli basis has all possible tensor products of Pauli matrices 
$\pauli{z}^{(j)}$ and identities $\mathbb{I}^{(k)}$ for 
$j,k=1,2,...\hat{i},...\,n$. Consequently, the measurement setting required 
to measure the expectation value of $K_{i}$ is composed of $n-1$ local 
measurements of type $Z$ and one measurement of kind $X$ to be performed on 
qubit $i$.\\

Consider now compositions of pairs of stabilizers, i.e. $K_{i}K_{j}$. 
Pauli matrices of type $\pauli{y}$ appear in positions $i$ and $j$ because 
of terms $\pauli{x}^{(i)}\pauli{z}^{(i)}=-i\pauli{y}^{(i)}$ and 
$\pauli{z}^{(j)}\pauli{x}^{(j)}=i\pauli{y}^{(j)}$. Given the hermiticity of 
the stabilizer operators and of their compositions, among all possible 
arising new terms, i.e. $\pauli{y}^{(i)}\otimes \pauli{y}^{(j)}$, 
$\pauli{x}^{(i)}\otimes \pauli{x}^{(j)}$, 
$i\pauli{x}^{(i)}\otimes \pauli{y}^{(j)}$ and 
$i\pauli{y}^{(i)}\otimes \pauli{x}^{(j)}$, only the first two appear. 
We notice that, due to the hermiticity requirement, only an even number of 
local operators of kind $Y$ appear. It follows that, in order to measure 
each composition of pair of stabilizers, 
two measurement settings are required: one with measurements of 
kind $X$ and one with measurements of kind $Y$  to be performed on qubits 
$i$ and $j$.\\

We consider now the composition of the pair $K_{i}K_{j}$ with stabilizer 
$K_{k}$. We have to consider first the new terms arising from the composition 
of $\pauli{x}^{(i)}\otimes \pauli{x}^{(j)}\otimes \pauli{z}^{(k)}$ and 
$\pauli{x}^{(i)} \otimes \pauli{x}^{(j)} \otimes \mathbb{I}^{(k)}$ with 
$\pauli{z}^{(i)}\otimes \pauli{z}^{(j)} \otimes \pauli{x}^{(k)}$, 
$\mathbb{I}^{(i)}\otimes \mathbb{I}^{(j)}\otimes \pauli{x}^{(k)}$, 
$\pauli{z}^{(i)}\otimes \mathbb{I}^{(j)}\otimes \pauli{x}^{(k)}$ and 
$\mathbb{I}^{(i)}\otimes \pauli{z}^{(j)} \otimes \pauli{x}^{(k)}$; the new 
terms are $8$ but, because of the hermiticity requirement, only half of them 
appear. We then consider those terms arising from the composition of 
$\pauli{y}^{(i)}\otimes \pauli{y}^{(j)}\otimes \pauli{z}^{(k)}$ and 
$\pauli{y}^{(i)} \otimes \pauli{y}^{(j)} \otimes \mathbb{I}^{(k)}$ with 
$\pauli{z}^{(i)}\otimes \pauli{z}^{(j)} \otimes \pauli{x}^{(k)}$, 
$\mathbb{I}^{(i)}\otimes \mathbb{I}^{(j)}\otimes \pauli{x}^{(k)}$, 
$\pauli{z}^{(i)}\otimes \mathbb{I}^{(j)}\otimes \pauli{x}^{(k)}$ and 
$\mathbb{I}^{(i)}\otimes \pauli{z}^{(j)} \otimes \pauli{x}^{(k)}$; they 
are $8$ but, because of the hermiticity requirement, only half of them appear.
The $4$ terms that survive in the first round are just the same as those 
that survive in the second; actually, the four terms generated in each round 
alone exhaust all admissible terms. In order to measure the operator
$K_{i}K_{j}K_{k}$, we conclude that $4$ local measurement settings are 
therefore required. In fact, due to the hermiticity requirement, only an even 
number of local measurements of kind $Y$ needs to be performed: the number 
of ways in which we can choose an even set of qubits among $3$, onto which 
perform local measurements of kind $Y$, is indeed exactly equal to $4$.\\

This reasoning may be extended to any compositions of stabilizer operators 
of the form $\prod_{i=1}^{k}K_{i}$ with $2\leq k\leq n$. Because of the 
hermiticity requirement, only an even number of local operators of kind $Y$ 
appears: the number of ways in which we can choose an even set of qubits 
among $k$, onto which perform local measurements of kind $Y$, is equal to 
$n_{k}$ where
\begin{displaymath}
	n_{k}=\sum_{k'=0,\,\text{even}}^{k} \binom{k}{k'}\equiv 2^{k-1}\,.
\end{displaymath}
It follows that $2^{k-1}$ is the number of local measurement settings 
required to measure the expectation value of the product of $k$ 
stabilizer \cite{fn}.
Moreover, regarding the number of local measurement settings required, 
the sole maximum-cardinality hyperedge case is the most demanding: the local 
decomposition of the control gate $C_{n-1}^{(1,2,...\hat{i},...\,n)}$ has all 
possible tensor products of Pauli matrices $\pauli{z}^{(j)}$ and identities 
$\mathbb{I}^{(k)}$ for $j,k=1,2,...\hat{i},...\,n$. This happens to the 
stabilizer operators of any hypergraph states endowed with a 
maximum-cardinality 
hyperedge. Any other $k$-qubit control gate with $k< n-1$ has in its 
decomposition a subset of the terms appearing in the representation of 
control gate $C_{n-1}^{(1,2,...\hat{i},...\,n)}$.\\

We want to point out that the vanishing of anti-hermitian terms and the 
survival of the hermitian ones may be explained also in the following way. 
Consider the composition of stabilizers $K_{i}$ and $K_{j}$, then an odd 
number of $\pauli{y}$ matrices may appear as a consequence either of 
composition $\pauli{x}^{(i)}\pauli{z}^{(i)}=-i\pauli{y}^{(i)}$ in position 
$i$ or of composition $\pauli{z}^{(j)}\pauli{x}^{(j)}=i\pauli{y}^{(j)}$ in 
position $j$. Because of the hermiticity requirement, terms 
$i\pauli{y}^{(i)}$ and $-i\pauli{y}^{(j)}$ should appear too. Terms with an 
even number of $\pauli{y}$ matrices arise from the following tensor products: 
$i\pauli{y}^{(i)} \otimes (-i)\pauli{y}^{(j)}=\pauli{y}^{(i)} \otimes 
\pauli{y}^{(j)}$ and $-i\pauli{y}^{(i)} \otimes i\pauli{y}^{(j)}=
\pauli{y}^{(i)} \otimes \pauli{y}^{(j)}$. This means that in the even case, 
despite the single $\pauli{y}$ matrices having opposite signs, their tensor 
product results in having the same sign.\\

As regards the projective witness $W_{n}$, we express the projector 
$\ket{G_{n}}\bra{G_{n}}$ as the sum of the identity plus $n$ other 
different contributions, each contribution made up of terms of the form 
$\prod_{i=1}^{j}K_{i}$ for $j=1,2,...\, n$, i.e.
\begin{displaymath}
	\begin{split}
		W_{n}&=\frac{2^{n-1}-1}{2^{n-1}}\mathbb{I}-\ket{G_{n}}
\bra{G_{n}}\\
		&=\frac{2^{n-1}-1}{2^{n-1}}\mathbb{I}-\prod_{i=1}^{n}
\frac{K_{i}+\mathbb{I}}{2}\\
		&=\frac{2^{n-1}-1}{2^{n-1}}\mathbb{I}+\\
		&-\frac{1}{2^{n}}\left(\mathbb{I}+\sum_{i=1}^{n}K_{i}
+\sum_{i< j=1}^{n}K_{i}K_{j}+...\,\prod_{i=1}^{n}K_{i} \right) \,.
	\end{split}
\end{displaymath}
Then the number of local measurement settings required by the witness 
$W_{n}$ is $\sum_{k=1}^{n}2^{k-1}\binom{n}{k}=\frac{3^{n}-1}{2}$. 
The witness $\tilde{W}_{n}$ instead, since it can be written as a sum of 
single stabilizers, needs exactly $n$ local measurement settings in order 
to be measured.\\

We conclude that, on the one hand, witness $\tilde{W}_{n}$ may always be 
measured efficiently: the number of local measurement settings required 
grows linearly with $n$. On the other hand, the number of local measurement 
settings required by the projective witness $W_{n}$, is not only strictly 
greater than the number required by the corresponding stabilizer witness 
$\tilde{W}_{n}$, but, in the worst case scenario, grows exponentially with 
the number of qubits $n$.

\paragraph{General case} In order to measure each stabilizer operator, one 
local measurement setting is needed but it may happen that some stabilizers 
require the same local measurement setting. Let for instance $K_{i}$ and 
$K_{j}$ be the stabilizer operators associated to a pair of qubits $i$ and 
$j$ not directly connected to each other. Then the local representation of 
the stabilizer $K_{i}$ over the Pauli basis does not contain any 
$\pauli{z}^{(j)}$ term in the $j$-th position but only identities 
$\mathbb{I}^{(j)}$ and the local representation of stabilizer $K_{j}$ does 
not contain any $\pauli{z}^{(i)}$ term in the $i$-th position but only 
identities $\mathbb{I}^{(i)}$. One local measurement setting, composed of two 
measurements of kind $X$ to be performed on qubits $i$ and $j$ and 
measurements of type $Z$ to be performed on the remaining qubits, is then 
enough to measure the expectation values of the two stabilizers.\\
When compositions of two or more stabilizers are concerned, Pauli matrices 
of kind $\pauli{y}$ appear and local measures of kind $Y$ are required. 
The connectedness hypothesis grants, for each stabilizer $K_{i}$, the 
existence of at least one stabilizer $K_{j}$, such that their composition 
$K_{i}K_{j}$ generates Pauli matrices of kind $\pauli{y}$ in positions $i$ 
and $j$.\\
We conclude that, on the one hand, witness $\tilde{W}_{n}$ needs a number of 
local measurement settings lower or equal to $n$. On the other hand, the 
number of local measurement settings required by witness $W_{n}$ is lower or 
equal to $\sum_{k=1}^{n}2^{k-1}\binom{n}{k}=\frac{3^{n}-1}{2}$ but strictly 
greater than the number of local measurement settings required by the 
witness $\tilde{W}_{n}$.


\begin{thebibliography}{99}
\bibitem{fn}{The sum of the binomial coefficient $\binom{k}{k'}$ 
over $k'$ even is equal to the sum over $k'$ odd ($0\leq k' \leq k$), which 
is equal to half of the total sum. This equality is a direct consequence of 
the following identity:\begin{displaymath}\sum_{k'=0}^{k}(-1)^{k'}
\binom{k}{k'}\equiv 0\,.\end{displaymath}}
\end{thebibliography}
\end{document}